\documentclass{article}

\usepackage{fullpage}

\usepackage[draft]{hyll-report}

\usepackage{lastpage}
\usepackage[T1]{fontenc}
\usepackage[utf8]{inputenc}

\usepackage{lmodern}

\usepackage{amsmath}
\usepackage{amssymb}
\usepackage{relsize}
\usepackage{xspace}

\def\eqdef{\overset{\mathrm{def}}{=}} 

\newcommand\proofsystem[1]{{\smaller\rmfamily\slshape #1}\xspace}
\newcommand\LF{\proofsystem{LF}}
\newcommand\HLF{\proofsystem{HLF}}
\newcommand\Spi{\proofsystem{S$\pi$}}
\newcommand\CTMC{\proofsystem{CTMC}}

\title{A Hybrid Linear Logic for Constrained Transition Systems
       \\ with Applications to Molecular Biology}

\author{
  Kaustuv Chaudhuri\\INRIA\\\texttt{kaustuv.chaudhuri@inria.fr}
  \and
  Jo\"{e}lle Despeyroux\\INRIA and CNRS, I3S\\\texttt{joelle.despeyroux@inria.fr}
}

\date{{\em Rapport de recherche INRIA-HAL nb} inria-00402942 --- \today\ --- \pageref{LastPage} pages}

\begin{document}

\maketitle

\begin{abstract}
  Linear implication can represent state transitions, but real transition
  systems operate under temporal, stochastic or probabilistic constraints that
  are not directly representable in ordinary linear logic. We propose a general
  modal extension of intuitionistic linear logic where logical truth is indexed
  by constraints and hybrid connectives combine constraint reasoning with
  logical reasoning. The logic has a focused cut-free sequent calculus that can
  be used to internalize the rules of particular constrained transition systems;
  we illustrate this with an adequate encoding of the synchronous stochastic
  pi-calculus.
  We also present some preliminary experiments of direct encoding of biological 
  systems in the logic.
\end{abstract}

\section{Introduction}
\label{sec:intro}

To reason about state transition systems, we need a logic of state. 
Linear logic~\cite{girard87tcs} is such a logic and has been successfully used
to model such diverse systems as 
process calculi, references and concurrency in programming languages, 
security protocols, multi-set rewriting, and graph algorithms.
Linear logic achieves this versatility by representing propositions as \emph{resources}
that are composed into elements of state using "tens", which can then be
transformed using the linear implication ("-o").  However, linear implication is
timeless: there is no way to correlate two concurrent transitions.
If resources have lifetimes and state changes have temporal, probabilistic or
stochastic \emph{constraints}, then the logic will allow inferences that may not
be realizable in the system being modelized. 
The need for formal reaosning in such constrained systems has led to the
creation of specialized formalisms such as Computation Tree Logic
(\proofsystem{CTL})\cite{Emerson95}, Continuous Stochastic Logic
(\proofsystem{CSL})~\cite{aziz00tcl} or Probabilistic CTL
(\proofsystem{PCTL})~\cite{hansson94fac}.
These approaches pay a considerable encoding overhead for the states and
transitions in exchange for the constraint reasoning not provided by linear logic.
A prominent alternative to the logical approach is to use a suitably enriched
process algebra; a short list of examples includes 
reversible CCS~\cite{danos03bc}, bioambients~\cite{regev04tcs}, 
brane calculi~\cite{cardelli03bc}, stochastic and probabilistic $\pi$-calculi, 
the PEPA algebra~\cite{hillston96book}, and the $\kappa$-calculus~\cite{danos04tcs}.
Each process algebra comes equipped with an underlying algebraic semantics which
is used to justify mechanistic abstractions of observed reality as processes. 
These abstractions are then animated by means of simulation and then
compared with the observations. 
Process calculi do not however completely fill the need for
{\em formal logical reasoning for constrained transition systems}. 
For example, there is no uniform process calculus to encode 
different stochastic process algebras\footnote{Stochastic process algebras are 
typical examples of the constrained transition systems we aim at formalizing.}. 

Note that logics like CSL or CTL are not such uniform languages either. These formalisms
are not \emph{logical frameworks}\footnote
{
Logical frameworks are uniform languages that allow to formally not only specify and analyse,
but also compare, or translate from one to the other, different systems, through their 
(adequate) encodings.
}: 
Encoding the stochastic $\pi$ calculus in CSL, for example, would be
inordinately complex because CSL does not provide any direct means of encoding
$\pi$-calculus dynamics such as the linear production and consumption of
messages in a synchronous interaction.
Actually CSL and CTL mainly aim at specifying properties of behaviors of constrained 
transition systems, not the systems themselves.

We propose a simple yet general method to add constraint reasoning to linear logic.
It is an old idea---\emph{labelled deduction}~\cite{simpson94phd} with 
\emph{hybrid} connectives~\cite{brauener06jal}---applied to a new domain. 
Precisely, we parameterize ordinary logical truth on a \emph{constraint domain}: 
"A @ w" stands for the truth of "A" under constraint "w". 
Only a basic monoidal structure is assumed about the constraints from a
proof-theoretic standpoint.
We then use the hybrid connectives of \emph{satisfaction} ("at") and
\emph{localization} ("dn") to perform generic symbolic reasoning on the
constraints at the propositional level.
We call the result \emph{hybrid linear logic} (\hyll); it has a generic cut-free
(but cut admitting) sequent calculus that can be strengthened with a focusing
restriction~\cite{andreoli92jlc} to obtain a normal form for proofs.
Any instance of \hyll that gives a semantic interpretation to the constraints
enjoys these proof-theoretic properties.

Focusing allows us to treat \hyll as a \emph{logical framework} 
for constrained transition systems.
Logical frameworks with hybrid connectives have been considered before; hybrid
\LF (\HLF), for example, is a generic mechanism to add many different kinds of
resource-awareness, including linearity, to ordinary \LF~\cite{reed06hylo}.
\HLF follows the usual \LF methodology of keeping the logic of the framework
minimal: its proof objects are $\beta$-normal $\eta$-long natural deduction
terms, but the equational theory of such terms is sensitive to permutative
equivalences~\cite{watkins03tr}.
With a focused sequent calculus, we have more direct access to a canonical
representation of proofs, so we can enrich the framework with any connectives
that obey the focusing discipline.
The representational adequacy \footnote
{Encodings -of a system or of a property of a system- in a logical framework
are always required to be {\it adequate} in a strong sense 
sometimes called {\it representational adequacy} and illustrated here in \secref{spi.adq}.
}
of an encoding in terms of (partial) focused
sequent derivations tends to be more straightforward than in a natural deduction
formulation.
We illustrate this by encoding the synchronous stochastic $\pi$-calculus (\Spi)
in \hyll using rate functions as constraints.

In addition to the novel stochastic component, our encoding of \Spi is a
conceptual improvement over other encodings of $\pi$-calculi in linear logic.
In particular, we perform a full propositional reflection of processes as
in~\cite{miller92welp}, but our encoding is first-order and adequate as
in~\cite{cervesato03tr}.
\hyll does not itself prescribe an operational semantics for the encoding of
processes; thus, bisimilarity in continuous time Markov chains (\CTMC) is not
the same as logical equivalence in stochastic \hyll, unlike in
\proofsystem{CSL}~\cite{desharmais03jlap}.
This is not a deficiency; rather, the \emph{combination} of focused \hyll proofs
and a proof search strategy tailored to a particular encoding is necessary to
produce faithful symbolic executions.
This exactly mirrors \Spi where it is the simulation rather than the transitions
in the process calculus that is shown to be faithful to the \CTMC
semantics~\cite{phillips04cmmb}.

This work has the following main contributions.
First is the logic \hyll itself and its associated proof-theory, which has a
very standard and well understood design in the Martin-Löf tradition.
Second, we show how to obtain many different instances of \hyll for particular
constraint domains because we only assume a basic monoidal structure for
constraints.
Third, we illustrate the use of focused sequent derivations to obtain adequate
encodings by giving a novel adequate encoding of \Spi.
Our encoding is, in fact, \emph{fully adequate}, \ie, partial focused proofs are
in bijection with traces.
The ability to encode \Spi gives an indication of the versatility of \hyll.
Finally, we show how to encode (a very simple example of) biological systems in \hyll.
This is a preliminary step
towards a logical framework for systems biology, our initial motivation for this work.

The sections of this paper are organized as follows: in \secref{hyll}, we
present the inference system (natural deduction and sequent calculus) for \hyll
and describe the two main semantic instances: temporal and probabilistic constraints.  
In \secref{focusing} we sketch the general focusing restriction on \hyll sequent
proofs. 
In \secref{spi} we give the encoding of \spi in probabilistic \hyll, and show that
the encoding is representationally adequate for focused proofs
(theorems \ref{thm:completeness} and \ref{thm:adeq}).  
In \secref{bio} we present some preliminary experiments of direct encoding of 
biological systems in \hyll.
We end with an overview of related (\secref{related}) and future work 
(\secref{concl}).

\section{Hybrid linear logic}
\label{sec:hyll}

In this section we define \hyll, a conservative extension of intuitionistic
first-order linear logic (\ill)~\cite{girard87tcs} where the truth judgements
are labelled by worlds representing constraints. Like in \ill, propositions are
interpreted as \emph{resources} which may be composed into a \emph{state} using
the usual linear connectives, and the linear implication ("-o") denotes a
transition between states. The world label of a judgement represents a
constraint on states and state transitions; particular choices for the worlds
produce particular instances of \hyll. The common component in all the instances
of \hyll is the proof theory, which we fix once and for all. We impose the
following minimal requirement on the kinds of constraints that \hyll can deal
with.

\begin{defn} \label{defn:constraint-domain}
  A \emph{constraint domain} "\cal W" is a monoid structure
  "\langle W, ., rid\rangle". The elements of "W" are called \emph{worlds}, and
  the partial order "\preceq\ : W \times W"---defined as "u \preceq w" if there
  exists "v \in W" such that "u . v = w"---is the \emph{reachability relation}
  in "\cal W".
\end{defn}

\noindent
The identity world "rid" is "\preceq"-initial and is intended to represent the
lack of any constraints. Thus, the ordinary \ill is embeddable into any instance
of \hyll by setting all world labels to the identity. When needed to
disambiguate, the instance of \hyll for the constraint domain "\cal W" will be
written \chyll[W].

The reader may wonder why we require worlds to be monoids, instead of lattices, for example.
The answer is to give a more general definition, suitable for rates constraints. 
One may then ask why we don't ask the worlds to be commutative. The answer is to allow 
lattices.

Atomic propositions are written using minuscule letters ("a, b, ...") applied to
a sequence of \emph{terms} ("s, t, \ldots"), which are drawn from an untyped
term language containing term variables ("x, y, \ldots") and function symbols
("f, g, ...") applied to a list of terms.. Non-atomic propositions are
constructed from the connectives of first-order intuitionistic linear logic and
the two hybrid connectives \emph{satisfaction} ("at"), which states that a
proposition is true at a given world ("w, u, v, \ldots"), and
\emph{localization} ("dn"), which binds a name for the (current) world the proposition is
true at. The following grammar summarizes the syntax of \hyll propositions.

\smallskip
\bgroup
\begin{tabular}{l@{\ }r@{\ }l}
  "A, B, ..." & "::=" & "a~\vec t OR A tens B OR one OR A -o B OR A with B OR top OR A plus B OR zero OR ! A OR all x. A OR ex x. A" \\ 
              & "|\ " & "(A at w) OR now u. A OR all u. A OR ex u. A" \\
\end{tabular}
\egroup

\smallskip
\noindent
Note that in the propositions "now u. A", "all u. A" and "ex u. A", the scope of
the world variable "u" is all the worlds occurring in "A". World variables
cannot be used in terms, and neither can term variables occur in worlds; this
restriction is important for the modular design of \hyll because it keeps purely
logical truth separate from constraint truth.  We let "\alpha" range over
variables of either kind.

The unrestricted connectives "and", "or", "imp", \etc of intuitionistic
(non-linear) logic can also be defined in terms of the linear connectives and
the exponential "!"  using any of the available embeddings of intuitionistic
logic into linear logic, such as Girard's embedding~\cite{girard87tcs}.

\subsection{Natural deduction for \hyll}

We start with the judgements from linear logic~\cite{girard87tcs} and enrich
them with a modal situated truth. We present the syntax of hybrid linear logic
in a natural deduction style, using Martin-L\"{o}f's principle of separating
judgements and logical connectives.  
Instead of the ordinary mathematical judgement ``"A" is true'', for a proposition $A$,
judgements of \hyll are of the form ``"A" is true at world "w"'', abbreviated as "A @ w". 
We use dyadic hypothetical derivations of
the form "\G ; \D \ "|-"\ C @ w" where "\G" and "\D" are sets of judgements of
the form "A @ w", with "\D" being moreover a \emph{multiset}.  "\G" is called
the \emph{unrestricted context}: its hypotheses can be consumed any number of
times.  "\D" is a \emph{linear context}: every hypothesis in it must be consumed
singly in the proof.
Note that in a judgement "A @ w" (as in a proposition "A at w"), $w$ can be 
any expression in $\cal W$, not only a variable.

The rules for the linear connectives are borrowed from \cite{chaudhuri03tr}
where they are discussed at length, so we omit a more thorough discussion here.
The rules for the first-order quantifiers are completely standard.  The
unrestricted context "\G" enjoys weakening and contraction; as usual, this is a
theorem that is attested by the inference rules of the logic, and we omit its
straightforward inductive proof. 
The notation $ A [ \tau / \alpha ]$ stands for the replacement of all free occurrences
of the variable $\alpha$ in $A$ with the expression $\tau$, avoiding capture. 
Note that the expressions in the rules are to be readen up to alpha-conversion.

\begin{thm}[structural properties] \label{thm:struct} \mbox{}
  \begin{ecom}
  \item If "\G ; \D |-nd C @ w", then "\G, \G' ; \D |-nd C @ w". (weakening)
  \item If "\G, A @ u, A @ u ; \D |-nd C @ w", then "\G, A @ u ; \D |-nd C @ w". (contraction)
  \end{ecom}
\end{thm}

\begin{figure*}[tp]
\framebox{ \begin{minipage}{0.9\textwidth}

\setlength{\parindent}{0pt}

\paragraph{Judgemental rules}
\begin{gather*}
  \I[hyp]{"\G ; A @ w |-nd A @ w"}
  \SP
  \I[hyp!]{"\G, A @ w ; . |-nd A @ w"}
\end{gather*}

\paragraph{Multiplicatives}
\begin{gather*}
  \I["tens I"]{"\G ; \D, \D' |-nd A tens B @ w"}
    {"\G ; \D |-nd A @ w" & "\G ; \D' |-nd B @ w"}
  \LSP
  \I["tens E"]{"\G ; \D, \D' |-nd C @ {w'}"}
    {\begin{array}{c}"\G ; \D |-nd A tens B @ w" \\ "\G ; \D', A @ w, B @ w |-nd C @ {w'}"\end{array}}
  \SP
  \\[1ex]
  \I["one I"]{"\G ; . |-nd one @ w"}
  \SP
  \I["one E"]{"\G ; \D, \D' |-nd C @ {w'}"}
    {"\G ; \D |-nd one @ w" & "\G ; \D' |-nd C @ {w'}"}
  \\[1ex]
  \I["{-o}I"]{"\G ; \D |-nd A -o B @ w"}
    {"\G ; \D, A @ w |-nd B @ w"}
  \SP
  \I["{-o}E"]{"\G ; \D, \D' |-nd B @ w"}
    {"\G ; \D |-nd A -o B @ w" & "\G ; \D' |-nd A @ w"}
\end{gather*}

\paragraph{Additives}
\begin{gather*}
  \I["with I"]{"\G ; \D |-nd A with B @ w"}
    {"\G ; \D |-nd A @ w" & "\G ; \D |-nd B @ w"}
  \SP
  \I["with E_i"]{"\G ; \D |-nd A_i @ w"}
    {"\G ; \D |-nd A_1 with A_2 @ w"}
  \\[1ex]
  \I["plus I_i"]{"\G ; \D |-nd A_1 plus A_2 @ w"}
    {"\G ; \D |-nd A_i @ w"}
  \SP
  \I["plus E"]{"\G ; \D, \D' |-nd C @ {w'}"}
    {"\G ; \D |-nd A plus B @ w"
     &
     \begin{array}[b]{c}
       "\G ; \D', A @ w |-nd C @ {w'}" \\
       "\G ; \D', B @ w |-nd C @ {w'}"
     \end{array}
    }
  \\[1ex]
  \I["top I"]{"\G ; \D |-nd top @ w"}
  \SP
  \I["zero E"]{"\G ; \D, \D' |-nd C @ {w'}"}{"\G ; \D |-nd zero @ w"}
\end{gather*}

\paragraph{Quantifiers}
\begin{gather*}
  \I["\forall I^\alpha"]{"\G ; \D |-nd \fall \alpha A @ w"}
    {"\G ; \D |-nd A @ w"}
  \SP
  \I["\forall E"]{"\G ; \D |-nd [\tau / x] A @ w"}
    {"\G ; \D |-nd \fall \alpha A @ w"}
  \\[1ex]
  \I["\exists I"]{"\G ; \D |-nd \fex \alpha A @ w"}
    {"\G ; \D |-nd [\tau / x] A @ w"}
  \SP
  \I["\exists E^\alpha"]{"\G ; \D, \D' |-nd C @ {w'}"}
    {"\G ; \D |-nd \fex \alpha A @ w" & "\G ; \D', A @ w |-nd C @ {w'}"}
\end{gather*}

\hfill
\begin{minipage}{0.7\linewidth}
  For "\forall I^\alpha" and "\exists E^\alpha", "\alpha" is assumed to be fresh
  with respect to the conclusion.\\
  For "\exists I" and "\forall E", "\tau" stands for a term or world, as
  appropriate.
\end{minipage}

\paragraph{Exponentials}
\begin{gather*}
  \I["! I"]{"\G ; . |-nd {! A} @ w"}
    {"\G ; . |-nd A @ w"}
  \SP
  \I["! E"]{"\G ; \D, \D' |-nd C @ {w'}"}
    {"\G ; \D |-nd {! A} @ w" & "\G, A @ w ; \D' |-nd C @ {w'}"}
\end{gather*}

\paragraph{Hybrid connectives}
\begin{gather*}
  \I["at I"]{"\G ; \D |-nd (A at w) @ {w'}"}
    {"\G ; \D |-nd A @ w"}
  \SP
  \I["at E"]{"\G ; \D |-nd A @ w"}
    {"\G ; \D |-nd (A at w) @ {w'}"}
  \\[1ex]
  \I["{dn} I"]{"\G ; \D |-nd now u. A @ w"}
    {"\G ; \D |-nd [w / u] A @ w"}
  \SP
  \I["{dn} E"]{"\G ; \D |-nd [w / u] A @ w"}
    {"\G ; \D |-nd now u. A @ w"}
\end{gather*}
\end{minipage}}
\caption{Natural deduction for \hyll}.
\label{fig:nd-rules}
\end{figure*}

The full collection of inference rules are in \figref{nd-rules}. A brief
discussion of the hybrid rules follows. To introduce the \emph{satisfaction}
proposition "(A at w)" (at any world "w'"), the proposition "A" must be true in the
world "w". The proposition "(A at w)" itself is then true at any world, not just
in the world "w". In other words, "(A at w)" carries with it the world at which
it is true. Therefore, suppose we know that "(A at w)" is true (at any world "w'");
then, we also know that "A @ w". These two introduction and
elimination rules match up precisely to (de)construct the information in the "A
@ w" judgement.
The other hybrid connective of \emph{localisation}, "dn", is intended to be
able to name the current world. That is, if "now u. A" is true at world "w",
then the variable "u" stands for "w" in the body "A". This interpretation is
reflected in its introduction rule "{dn} I". For elimination, suppose we have a
proof of "now u. A @ w" for some world "w". Then, we also know "[w / u] A @ w".

For the linear and unrestricted hypotheses, substitution is no different 
from that of the usual linear logic.

\begin{thm}[substitution] \label{thm:subst} \mbox{}
  \begin{ecom}
  \item \label{thm:subst.1}
    If "\G ; \D |-nd A @ u" and "\G ; \D', A @ u |-nd C @ w", then "\G ; \D, \D' |-nd C @ w".
  \item  \label{thm:subst.2}
    If "\G ; . |-nd A @ u" and "\G, A @ u ; \D |-nd C @ w", then "\G ; \D |-nd C @ w".
  \end{ecom}
\end{thm}

\begin{proof}[Proof sketch]
  By structural induction on the second given derivation in each case.
\end{proof}

Note that the "dn" connective commutes with every propositional connective, including itself. That
is, "now u. (A * B)" is equivalent to "(now u. A) * (now u. B)" for all binary connectives "*", and
"now u. * A" is equivalent to "* (now u. A)" for every unary connective "*", assuming the
commutation will not cause an unsound capture of "u". It is purely a matter of taste where to place
the "dn", and repetitions are harmless.

\begin{thm}[conservativity]
  \label{thm:conserv}
  Call a proposition or multiset of propositions \emph{pure} if it contains no
  instance of the hybrid connectives and no instance of quantification over a world variable, 
  and let "\G", "\D" and "A" be pure. 
  Then, "\G ; \D |-nd A @ w" in \hyll iff "\G ; \D |-nd A" in intuitionistic linear logic.
\end{thm}

\begin{proof}
  By structural induction on the given \hyll derivation.
\end{proof}

\subsection{Sequent calculus for \hyll}

In this section, we give a sequent calculus presentation of \hyll and prove a
cut-admissibility theorem.  The sequent formulation in turn will lead to an
analysis of the polarities of the connectives in order to get a focused sequent
calculus that can be used to compile a logical theory into a system of derived
inference rules with nice properties (\secref{focusing}).  For instance, if a
given theory defines a transition system, then the derived rules of the focused
calculus will exactly exhibit the same transitions. This is key to obtain the
necessary representational adequacy theorems,
as we shall see for the \spi-calculus example chosen in this paper (\secref{spi.adq}). 

In the sequent calculus, we depart from the linear hypothetical judgement "|-"
which has only an ``active'' right-hand side to a sequent arrow "==>" that has
active zones on both sides. A rule that infers a proposition on the right of the
sequent arrow is called a ``right'' rule, and corresponds exactly to the
introduction rules of natural deduction. Dually, introductions on the left of
the sequent arrow correspond to elimination rules of natural deduction; however,
as all rules in the sequent calculus are introduction rules, the information
flow in a sequent derivation is always in the same direction: from the
conclusion to the premises, incidentally making the sequent calculus ideally
suited for proof-search.

\begin{figure}[p]
\framebox{ 
\begin{minipage}{0.9\textwidth}
  \setlength{\parindent}{0pt}
   \bgroup \small
   \paragraph{Judgemental rules}
   \begin{gather*}
     \I[init]{"\G ; a\ \vec t\ @ u ==> a\ \vec t\ @ u"}{}
     \SP
     \I[copy]{"\G, A @ u ; \D ==> C @ w"}{"\G, A @ u ; \D, A @ u ==> C @ w"}
   \end{gather*}

   \paragraph{Multiplicatives}
   \begin{gather*}
     \I["{tens}R"]{"\G ; \D, \D' ==> A tens B @ w"}{"\G ; \D ==> A @ w" & "\G ; \D' ==> B @ w"}
     \SP
     \I["{tens}L"]{"\G ; \D, A tens B @ u ==> C @ w"}
     {"\G ; \D, A @ u, B @ u ==> C @ w"}
     \\[1ex]
     \I["{one}R"]{"\G ; . ==> one @ w"}
     \SP
     \I["{one}L"]{"\G ; \D, one @ u ==> C @ w"}{"\G ; \D ==> C @ w"}
     \SP
     \I["{-o}R"]{"\G ; \D ==> A -o B @ w"}{"\G ; \D, A @ w ==> B @ w"}
     \\[1ex]
     \I["{-o}L"]{"\G ; \D, \D', A -o B @ u ==> C @ w"}
     {"\G ; \D ==> A @ u" & "\G ; \D', B @ u ==> C @ w"}
   \end{gather*}

   \paragraph{Additives}
   \begin{gather*}
     \I["top R"]{"\G ; \D ==> top @ w"}
     \LSP
     \I["{with}R"]{"\G ; \D ==> A with B @ w"}{"\G ; \D ==> A @ w" & "\G ; \D ==> B @ w"}
     \\[1ex]
     \I["{with}L_i"]{"\G ; \D, \D', A_1 with A_2 @ u ==> C @ w"}
     {"\G ; \D, A_i @ u ==> C @ w"}
     \\[1ex]
     \I["{plus}R_i"]{"\G ; \D ==> A_1 plus A_2 @ w"}{"\G ; \D ==> A_i @ w"}
     \LSP
     \I["zero L"]{"\G ; \D, zero @ u ==> C @ w"}
     \\[1ex]
     \I["{plus}L"]{"\G ; \D, A plus B @ u ==> C @ w"}
     {"\G ; \D, A @ u ==> C @ w" & "\G ; \D, B @ u ==> C @ w"}
   \end{gather*}

   \paragraph{Quantifiers}
   \begin{gather*}
     \I["\forall R^\alpha"]{"\G ; \D ==> \fall \alpha A @ w"}{"\G ; \D ==> A @ w"}
     \SP
     \I["\forall L"]{"\G ; \D, \fall \alpha A @ u ==> C @ w"}
     {"\G ; \D, [\tau / \alpha] A @ u ==> C @ w"}
     \\[1ex]
     \I["\exists R"]{"\G ; \D ==> \fex \alpha A @ w"}{"\G ; \D ==> [\tau / \alpha] A @ w"}
     \SP
     \I["\exists L^\alpha"]{"\G ; \D, \fex \alpha A @ u ==> C @ w"}
     {"\G ; \D, A @ u ==> C @ w"}
   \end{gather*}

   For "\forall R^\alpha" and "\exists L^\alpha", "\alpha" is assumed to
   be fresh with respect to the conclusion. For "\exists R" and "\forall
   L", "\tau" stands for a term or world, as appropriate.

   \paragraph{Exponentials}
   \begin{gather*}
     \I["{!}R"]{"\G ; . ==> {! A} @ w"}{"\G ; . ==> A @ w"}
     \SP
     \I["{!}L"]{"\G ; \D, {! A} @ u ==> C @ w"}{"\G, A @ u ; \D ==> C @ w"}
   \end{gather*}

   \paragraph{Hybrid connectives}

   \begin{gather*}
     \I["at R"]{"\G ; \D ==> (A at u) @ v"}{"\G ; \D ==> A @ u"}
     \LSP
     \I["at L"]{"\G ; \D, (A at u) @ v ==> C @ w"}{"\G ; \D, A @ u ==> C @ w"}
     \\[1ex]
     \I["{dn}R"]{"\G ; \D ==> now u. A @ w"}{"\G ; \D ==> [w/u] A @ w"}
     \LSP
     \I["{dn}L"]{"\G ; \D, now u. A @ v ==> C @ w"}{"\G ; \D, [v/u] A @ v ==> C @ w"}
   \end{gather*}
   \egroup
\end{minipage}}
\caption{The sequent calculus for \hyll}.
\label{fig:seq-rules}
\end{figure}

The full collection of rules of the \hyll sequent calculus is in
\figref{seq-rules}. There are only two structural rules: the init rule infers an
atomic initial sequent, and the copy rule introduces a contracted copy of an
unrestricted assumption into the linear context (reading from conclusion to
premise). Weakening and contraction are admissible rules:

\begin{thm}[structural properties] \mbox{}
  \begin{ecom}
  \item If "\G ; \D ==> C @ {w}", then "\G, \G' ; \D ==> C @ {w}". (weakening)
  \item If "\G, A @ u, A @ u ; \D ==> C @ {w}", then "\G, A @ u ; \D ==> C @ {w}". (contraction)
  \end{ecom}
\end{thm}

\begin{proof}
  By straightforward structural induction on the given derivations.
\end{proof}

The most important structural properties are the admissibility of the identity
and the cut principles. The identity theorem is the general case of the init
rule and serves as a global syntactic completeness theorem for the
logic. Dually, the cut theorem below establishes the syntactic soundness of the
calculus; moreover there is no cut-free derivation of ". ; . ==> zero @ w", so
the logic is also globally consistent.

\begin{thm}[identity] 
  "\G ; A @ w ==> A @ w".
\end{thm}

\begin{proof}
  By induction on the structure of "A" (see \secref{proofs.identity}).
\end{proof}

\bgroup
\begin{thm}[cut] \mbox{}
  \begin{ecom}[1.]
  \item If "\G ; \D ==> A @ u" and "\G ; \D', A @ u ==> C @ w", then "\G ; \D,
    \D' ==> C @ w".
  \item If "\G ; . ==> A @ u" and "\G, A @ u ; \D ==> C @ w", then "\G ; \D
    ==> C @ w".
  \end{ecom}
\end{thm}
\egroup

\begin{proof}
  By lexicographic structural induction on the given derivations, with cuts of
  kind 2 additionally allowed to justify cuts of kind 1. The style of proof
  sometimes goes by the name of \emph{structural
    cut-elimination}~\cite{chaudhuri03tr}.
  See \secref{proofs.cut} for the details.
\end{proof}

We can use the admissible cut rules to show that the following rules are
invertible: "tens L", "one L", "plus L", "zero L", "\exists L", "-o R", "with
R", "top R", and "\forall R". In addition, the four hybrid rules, "at R", "at
L", "{dn} R" and "{dn} L" are invertible. In fact, "dn" and "at" commute freely
with all non-hybrid connectives:

\begin{thm}[Invertibility] The following rules are invertible:
  \begin{ecom}
  \item On the right: "with R", "top R", "{-o} R", "\forall R", "{dn} R" and "at R";
  \item On the left: "tens L", "one L", "plus L", "zero L", "\exists L", "!L", "{dn} L" and "at L".
  \end{ecom}
\end{thm}

\begin{proof}
  See \S\ref{sec:proofs.invert}.
\end{proof}

\begin{thm}[Correctness of the sequent calculus] \mbox{}
  \begin{ecom}
  \item If "\G ; \D ==> C @ w", then "\G ; \D |-nd C @ w". (soundness)
  \item If "\G ; \D |-nd C @ w", then "\G ; \D ==> C @ w". (completeness)
  \end{ecom}
\end{thm}

\begin{proof}
  See \S\ref{sec:proofs.correct}.
\end{proof}

\begin{cor}[consistency]
  There is no proof of ". ; . |-nd zero @ w".
\end{cor}

\begin{proof}
  See \S\ref{sec:proofs.correct}.
\end{proof}

\hyll is conservative with respect to ordinary intuitionistic logic: as long as
no hybrid connectives are used, the proofs in \hyll are identical to those in
\ill~\cite{chaudhuri03tr}. The proof (omitted) is by simple structural
induction.

\begin{thm}[conservativity]
  If "\G ; \D ==>_{\hyll} C @ w" is derivable, contains no occurrence of the
  hybrid connectives "dn", "at", "\forall u" or "\exists u", and each element of
  "\G" and "\D" is of the form "A @ w", then "\G ; \D ==>_{\ill} C". 
\end{thm}

An example of derived statements, true in every semantics for worlds, is the following:
\begin{proposition}[relocalisation]
\label{thm:relocalisation}
\begin{gather*}
  \Ic[]{"\G ; A_1 @ {u . w_1} \cdots A_k @ {u . w_k} |- B @ {u . v}"}
         {"\G ; A_1 @ w_1 \cdots  A_k @ w_k |- B @ v"}
\end{gather*}
\end{proposition}
This property is particularly well suited to applications in biology.

\paragraph{}
In the rest of this paper we use the following derived connectives:

\bgroup 
\begin{defn}[modal connectives] \label{defn:connectives} \mbox{}
  \vspace{-1ex}
  \begin{gather*}
    \begin{aligned}
      "box A" &\triangleq "now u. all w. (A at u . w)" & \qquad
      "dia A" &\triangleq "now u. ex w. (A at u . w)" \\
      "rate v A" &\triangleq "now u. (A at u . v)" &
      "!! A" &\triangleq "all u. (A at u)"
    \end{aligned}
  \end{gather*}
\end{defn}
\egroup

\noindent
The connective "\rhoup" represents a form of delay. Note its derived right rule:
\begin{gather*}
  \Ic["\rhoup R"]{"\G ; \D |- rate v A @ w"}{"\G ; \D |- A @ {w . v}"}
\end{gather*}
The proposition "rate v A" thus stands for an \emph{intermediate state} in a
transition to "A". Informally it can be thought to be ``"v" before "A"''; thus,
"all v. rate v A" represents \emph{all} intermediate states in the path to "A",
and "ex v. rate v A" represents \emph{some} such state.  The modally
unrestricted proposition "!! A" represents a resource that is consumable in any
world; it is mainly used to make transition rules applicable at all worlds.

It is worth remarking that \hyll proof theory can be seen as at least as
powerful as (the linear restriction of) intuitionistic S5~\cite{simpson94phd}:

\begin{thm}[\hyll is S5]
  The following sequent is derivable: ". ; dia A @ w ==> box dia A @ w".
\end{thm}

\begin{proof}
  See \S\ref{sec:proofs.is5}.
\end{proof}
Obviously \hyll is more expressive as it allows direct manipulation of the worlds using
the hybrid connectives: for example, the $\rhoup$ connective is not definable in S5.

Let us elaborate a little bit on this point, together with the natural concerns of
allowing predicates on worlds in the propositions.

First of all, let us note that allowing predicates on worlds anywhere in the propositions
would yield the following false rule:
\begin{gather*}
  \Ic[]{"\G ; \D, \D' \not\vdash (now u. (u \neq w) tens A) @ w"}
       {"\G ; \D |- (w \neq w) @ w" \quad "\G ; \D' |- A @ w"}
\end{gather*}

It seems that allowing predicates on worlds in the propositions is only possible 
in a restricted form, by adding constrained conjunction and implication 
as done in CILL~\cite{saranli07icra} or $\eta$~\cite{deyoung08csf}. 
Following these works, we could allow the following expressions in the propositions:

\smallskip
\bgroup
\begin{tabular}{l@{\ }r@{\ }l}
  "A, B, ..." & "::=" & " ... OR (! wp)  tens A OR (! wp) -o B " \\ 
\end{tabular}
\egroup

where  "wp " is any predicate on worlds such as "w \neq w' " or "w \le w' ", for example. 

We might have chosen to define a modal connective "at'", instead of "at", with the following 
rules:
   \begin{gather*}
     \I["at' R"]{"\G ; \D ==> (A at' u) @ v"}{"\G ; \D ==> A @ u" \quad "(w \neq u)"}
     \LSP
     \I["at' L"]{"\G ; \D, (A at' u) @ v ==> C @ w"}
                {"\G ; \D, A @ u ==> C @ w" \quad "(w \neq u)"}
   \end{gather*}

Remarks:

1. If worlds are just monoid (not group), 
then S4 can be encoded in \hyll extended with constrained implication as above.

2. If worlds are groups 
(i.e. $...$ "." admits an inverse, 
i.e. $...$ W is a right cumulative magma),
then S5 can be encoded in \hyll with the modal connective "at'" defined above, instead of "at".

3.  If worlds are Kripke frames (i.e. total and symetric) then the relation $\le$ on worlds
can be internalized by a "at R" rule.

As Alex Simpson proved in his PhD thesis the cut elimination theorem for any intuitionistic
modal logic, we can be confident that the cut elimination theorem can be proven for 
\hyll with the modal connective "at'" instead of "at".

\subsection{Temporal constraints}
\label{sec:hyllt}

As a pedagogical example, consider the constraint domain "\mathcal{T} = \langle
\Reals+, +, 0\rangle" representing instants of time. This domain can be used to
define the lifetime of resources, such as keys, sessions, or delegations of
authority. Delay (\defnref{connectives}) in \hyllt represents intervals of time;
"rate d A" means ``"A" will become available after delay "d"'', similar to
metric tense logic~\cite{prior57book}. This domain is very permissive because
addition is commutative, resulting in the equivalence of "rate u rate v A" and
"rate v rate u A".
The ``forward-looking'' connectives "G" and "F" of ordinary tense logic are
precisely "box" and "dia" of \defnref{connectives}. 

In addition to the
future connectives, the domain "\mathcal{T}" also admits past connectives if we add
saturating subtraction (\ie, "a - b = 0" if "b \ge a") to the language of
worlds. We can then define the duals "H" and "P" of "G" and "F" as:
$$ "H~A" \triangleq "now u. all w. (A at u - w)" \qquad
   "P~A" \triangleq "now u. ex w. (A at u - w)"
$$
While this domain does not have any branching structure like CTL, it is
expressive enough for many common idioms because of the branching structure
of derivations involving $\oplus$. CTL reachability (``in some path in some
future''), for instance, is the same as our "dia"; similarly
CTL steadiness (``in some path for all futures'') is the same as $\Box$.
CTL stability (``in all paths in all futures''), however, has no direct correspondance in HyLL.
Note that model checking cannot cope with temporal expresssions involving the 
``in all path'' notion anyway. 
Thus approaches using ordinary temporal logics and model checking, 
like BIOCHAM, for example, cannot deal with those expressions either.

On the other hand, the availability of linear reasoning makes certain kinds
of reasoning in \hyll much more natural than in ordinary temporal
logics. One important example is of \emph{oscillation} between states in
systems with kinetic feedback. In a temporal specification language such as
BIOCHAM~\cite{chabrier05cmsb}, only finite oscillations are representable
using a nested syntax, while in \hyll we use a simple bi-implication; for
example, the oscillation between "A" and "B" with delay "d" is represented
by the rule "!!  (A -o rate d B) with (B -o rate d A)" (or "!! (A -o dia B)
with (B -o dia A)" if the oscillation is aperiodic).
If \hyllt were extended with constrained implication and conjunction in the
style of CILL~\cite{saranli07icra} or $\eta$~\cite{deyoung08csf}, then we can
define localized versions of "box" and "dia", such as ``"A" is true
everywhere/somewhere in an interval''. They would also allow us to define the
``until'' and ``since'' operators of linear temporal logic~\cite{kamp68phd}.

\subsection{Probabilistic Constraints}
\label{sec:hyllp}

The material in this section requires some background in probability and measure theory,
and can be skipped at a first reading, without significant loss of continuity.

Transitions in practice rarely have precise delays. Phenomenological and
experimental evidence is used to construct a probabilistic model of the
transition system where the delays are specified as probability distributions of
continuous variables. 

The meaning of the random variables depends on the intended application.
In the applications in the area of systems biology, 
the variables $X$ can represent the concentration of a product, 
while in economics, $X$ could be the duration of an activity, for example.

\subsection{General Case}

Let us recall some basic definitions in probability theory.
The probability of $X$ being in $A$ is defined by:
$\mu_{X}(A) = \int_{x \in A} \mu_X dx$.
For example:
P${rob}(X \le x)= \mu_{X}(x) = \int_{-\infty}^{x} \mu_X dt$.

\begin{fact}[see \cite{Foata-Fuchs-french-book,Rogers-Williams-vol1-book}]
  \label{fact:convolution}
If $X$ and $Y$ are independent random variables in $\mathbb{R}$, with distribution $\mu_X$ and $\mu_Y$, respectively, then the distribution $\mu_{X+Y}$ of the random variable $X+Y$ is given by
$$
\mu_{X+Y}(A)=\mu_X*\mu_Y(A)=\int_{\{x+y\in A\}}\mu_X(dx)\otimes\mu_Y(dy)
$$
for all Borel
\footnote{The set of Borel events is the set of the Lebesgue measurable functions.}
 subset $A$ of $\mathbb{R}$. 
\end{fact}

The space of probability distributions of random variables, together with the convolution operator $*$ and the Dirac mass at 0 ($\delta_0$)\footnote{
The Dirac mass at 0 $\delta_0$ is not a real function but a generalized function. 
It is just a measure.
},
as neutral element, forms a monoid. More precisely:

\begin{defn}
\label{defn:proba}
  The \emph{probabilities domain} $\mathcal{P}$ is the monoid 
  $\langle \mathbf{M_1}(\mathbb{R}),\ast, \delta_0 \rangle$ where 
  $\mathbf{M_1}(\mathbb{R})$ is the set of the 
  Borel probability measures over $\mathbb{R}$ and $\delta_0$ is Dirac mass at 0.
  The instance HyLL[$\mathcal{P}$] will sometimes be called 
  ``probabilistic hybrid linear logic''.
\end{defn}

An element $w= \mu_X(x)$ of a world $\mathcal{P}$ thus represents the probability of 
$X$ to have its value in the interval $[- \infty, x]$.
$\bar{w} = 1-\mu_X(x)$ represents the probability of $X$ to have its value greater than 
$x$.
"A @ \bar{w}" therefore means `$A$ is true with probability greater than $x$''.

\subsection{Markov Processes}

The standard model of stochastic transition systems is continuous time Markov chains (CTMCs)
where the delays of transitions between states are distributed according to the Markov 
assumption of memorylessness (Markov processes)
with the further condition that their state-space are countable sets 
\cite{Rogers-Williams-vol1-book}.

\begin{fact}[see \cite{Ethier-Kurtz-book,Rogers-Williams-vol1-book}]
  \label{fact:feller}
Given a continuous-time Markov process $(X_t,t\geq 0)$ taking values in a measurable space $(E,{\cal E})$, the family $(P(t),t\geq 0)$ of linear operators on the set of bounded Borel functions ${\cal B}(E)$ defined by: for all $f\in{\cal B}(E)$ and for all $x\in E$,
$$
(P(t)f)(x)  
           = \mathbf{E}[f(X_t)\mid X_0=x],
$$
where ${\mathbf E}$ is the expectation function,
is a semigroup for the convolution: for all $s,t\geq 0$,
$$
P(t+s) = P(t) * P(s)
$$
with neutral element $P(0)$, the identity operator. When the process $X$ is a Feller process (see \cite{Rogers-Williams-vol1-book}, chapter 3, section 2), $(P(t),t\geq 0)$ is a Feller semigroup, i.e. strongly continuous and conservative.
\end{fact}

For a given continuous-time Markov process $(X_t,t\geq 0)$, the associated monoid is the 
set $(P(t),t\geq 0)$ defined above. More precisely, we can define the Markov domain as follows:
  
\begin{defn}
\label{defn:markov}
  For a given continuous-time Markov process $(X_t,t\geq 0)$,  taking values in a 
  measurable space $(E,{\cal E})$, the \emph{Markov domain} $\mathcal{M}$ is
  the monoid $\langle (P(t),t\geq 0), \ast, P(0) \rangle$ where 
  $(P(t),t\geq 0)$ is the sub-markov semigroup of linear operators on the set of bounded
  Borel functions ${\cal B}(E)$ defined by for all $f\in{\cal B}(E)$ and for all $x\in E$,
  $(P(t)f)(x)=\mathbf{E}[f(X_t)\mid X_0=x]$.
  The instance HyLL($\mathcal{M}$) will sometimes be called 
  ``Markov hybrid linear logic''.
\end{defn}

In the above definition, $f$ can be any function in ${\cal B}(E)$, and we have
\begin{gather*}
(P(t)f)(x) = {\mathbf E}[f(X_t)\mid X_0=x].
\end{gather*}
An element $w= P(t)$ of $\mathcal{M}$ represents a function which
associates to any function $f$ (where $f \in{\cal B}(E))$ and to any initial
value $x$ for the variable $X_t$, the expectation of $f(X_t)$, knowing that $X_0=x$.
We can choose $f = {\mathbf 1}_A$: the indicator (i.e. characteristic) function of a set $A$.
In this case (recalling that ${\mathbf E}({\mathbf 1}_A) = P(A)$),
\begin{gather*}
(P(t)f)(x) = (P(t){\mathbf 1}_A)(x)
= {\mathbf E}[{\mathbf 1}_A(X_t)\mid X_0=x]
= P\{X_t \in A \mid X_0=x\}
\end{gather*}
For example, for $A=[- \infty, y]$:
$(P(t) {\mathbf 1}_A)(x) = P \{X_t \le y \mid X_0=x\} = {\mathbf F}_{X_t \mid X_0=x}(y)$
where ${\mathbf F}$ is the cumulative distribution function of $X_t$.
Other interesting examples for $f$ are the square function $\text{sq}(y) = y^2$
and the identity function $\text{id}(y) = y$. Using these functions, we can
define the variance of $X_t$:
\begin{gather*}
(P(t)~ \text{sq})(x) - (P(t)~ \text{id})^2(x) = {\mathbf E}(X_t^2 \mid X_0=x) -
({\mathbf E}(X_t \mid X_0=x))^2
  = \mathbf{Var} (X_t \mid X_0=x)
\end{gather*}
In principle, using suitable functions $f$, we should be able to define
\emph{any descriptors of the probability distribution of our variable $X_t$.}

The meaning of "A @ w" varies depending on the choice for the function $f$. For
example, we have seen that in the case of $f = {\mathbf 1}_{[-\infty,y]}$, ~$P(t)
f (x) = {\mathbf F}_{X_t \mid X_0=x}(y)$. In this case, "A @ w" means ``$A$ is
true with probability less than $y$ at time $t$''. For $\bar{w}= 1 - P(t)$, "A @
\bar{w}" means ``$A$ is true with probability greater than $y$ at time $t$''.

\paragraph{Rates}
The cumulative distributions of the continuous time Markov chains (CTMCs) 
used in \spi are exponential \cite{phillips06tcsb}.
More precisely:
$\text{Prob}(X_t \le x + rt \mid X_0=x) = {\mathbf F}_{X_t \mid X_0=x}(x + rt) = 1 - e^{rt}$,
where $r$ ({\emph rates}) are functions depending on the time $t$.
In this case, we can use $f = {\mathbf 1}_{[-\infty, x+rt]}$.
However, it is simpler to 
work in $\hyll(\mathcal P)$ (\defnref{proba}) 
and define the worlds by particularizing the general case of probabilities domains 
to the case where the cumulative distributions of all variables $X$ are exponential 
with the above meaning.
The worlds $w$ will therefore be defined by
$w = \mu_X(x + rt) = \text{Prob}(X_t \le x + rt \mid X_0=x) 
= {\mathbf F}_{X_t \mid X_0=x}(x + rt) = 1 - e^{rt}$.

\section{Focusing}
\label{sec:focusing}

As \hyll is intended to represent transition systems adequately, it is crucial
that \hyll derivations in the image of an encoding have corresponding
transitions. However, transition systems are generally specified as rewrite
algebras over an underlying congruence relation. These congruences have to be
encoded propositionally in \hyll, so a \hyll derivation will generally require
several inference rules to implement a single transition; moreover, several
trivially different reorderings of these ``micro'' inferences would correspond
to the same transition. It is therefore futile to attempt to define an
operational semantics directly on \hyll inferences.

We restrict the syntax to focused derivations~\cite{andreoli92jlc}, which
ignores many irrelevant rule permutations in a sequent proof and divides the
proof into clear \emph{phases} that define the grain of atomicity. The logical
connectives are divided into two classes, \emph{negative} and \emph{positive},
and rule permutations for connectives of like polarity are confined to
\emph{phases}. A \emph{focused derivation} is one in which the positive and
negative rules are applied in alternate maximal phases in the following way: in
the \emph{active} phase, all negative rules are applied (in irrelevant order)
until no further negative rule can apply; the phase then switches and one
positive proposition is selected for \emph{focus}; this focused proposition is
decomposed under focus (\ie, the focus persists unto its sub-propositions) until
it becomes negative, and the phase switches again.

As noted before, the logical rules of the hybrid connectives "at" and "dn" are
invertible, so they can be considered to have both polarities. It would be valid
to decide a polarity for each occurrence of each hybrid connective
independently; however, as they are mainly intended for book-keeping during
logical reasoning, we define the polarity of these connectives in the following
\emph{parasitic} form: if its immediate subformula is positive (resp. negative)
connective, then it is itself positive (resp. negative). These connectives
therefore become invisible to focusing. This choice of polarity can be seen as a
particular instance of a general scheme that divides the "dn" and "at"
connectives into two polarized forms each. To complete the picture, we also
assign a polarity for the atomic propositions; this restricts the shape of
focusing phases further~\cite{chaudhuri08jar}.
The full syntax of positive ($P, Q, \ldots$) and negative ($M, N, \ldots$)
propositions is as follows:

\medskip
\bgroup 
\hspace{-1.5em}
\begin{tabular}{l@{\ }r@{\ \ }l}
  "P, Q, ..." & "::=" & "p~\vec t OR P tens Q OR one OR P plus Q OR zero OR {! N} OR \fex \alpha P OR now u. P OR (P at w) OR pos N" \\
  "N, M, ..." & "::=" & "n~\vec t OR N with N OR top OR P -o N OR \fall \alpha N OR now u. N OR (N at w) OR neg P"
\end{tabular}
\egroup

\begin{figure}[p]
\hspace{-3em}
\framebox{ 
\begin{minipage}{1.1\textwidth}
  \setlength{\parindent}{0pt}
  \bgroup \small
  \paragraph{Focused logical rules}
  \begin{gather*}
    \I[li]{"\G ; foc{n\ \vec t @ w} ==> pos {n\ \vec t @ w}"}
    \quad
    \I["neg L"]{"\G ; \D ; foc{neg P @ u} ==> Q @ w"}{"\G ; \D ; P @ u ==> . ; Q @ w"}
    \qquad
    \I["with L_i"]{"\G ; \D ; foc{N_1 with N_2 @ u} ==> Q @ w"}{"\G ; \D ; foc{N_i @ u} ==> Q @ w"}
    \\[1ex]
    \I["{-o}L"]{"\G ; \D, \X ; foc{P -o N @ u} ==> Q @ w"}{
      "\G ; \D ==> foc{P @ u}" & "\G ; \X ; foc{N @ u} ==> Q @ w"
    }
    \quad
    \I["\forall L"]{"\G ; \D ; foc {\fall \alpha N @ u} ==> Q @ w"}{
      "\G ; \D ; foc {[\tau / \alpha] N @ u} ==> Q @ w"
    }
    \\[1ex]
    \I["{dn}LF"]{"\G ; \D ; foc{now u. N @ v} ==> Q @ w"}{"\G ; \D ; foc{[v / u] N @ v} ==> Q @ w"}
    \quad
    \I["{at}LF"]{"\G ; \D ; foc{(N at u) @ v} ==> Q @ w"}{"\G ; \D ; foc{N @ u} ==> Q @ w"}
    \\[1ex]
    \I[ri]{"\G ; neg {p\ \vec t} @ w ==> foc {p\ \vec t @ w}"}
    \quad
    \I["pos R"]{"\G ; \D ==> foc{pos N @ w}"}{"\G ; \D ; . ==> N @ w ; ."}
    \qquad
    \I["tens R"]{"\G ; \D, \X ==> foc{P tens Q @ w}"}{
      "\G ; \D ==> foc{P @ w}" & "\G ; \X ==> foc{Q @ w}"
    }
    \\[1ex]
    \I["one R"]{"\G ; . ==> foc {one @ w}"}
    \quad
    \I["plus R_i"]{"\G ; \D ==> foc{P_1 plus P_2 @ w}"}{"\G ; \D ==> foc{P_i @ w}"}
    \quad
    \I["!R"]{"\G ; . ==> foc{{!N} @ w}"}{"\G ; . ; . ==> N @ w ; ."}
    \\[1ex]
    \I["\exists R"]{"\G ; \D ==> foc{\fex \alpha P @ w}"}{"\G ; \D ==> foc{[\tau/\alpha] P @ w}"}
    \quad
    \I["{dn}RF"]{"\G ; \D ==> foc{now u. P @ w}"}{"\G ; \D ==> foc{[w/u] P @ w}"}
    \quad
    \I["at RF"]{"\G ; \D ==> foc{(P at u) @ w}"}{"\G ; \D ==> foc{P @ u}"}
  \end{gather*}

  \paragraph{Active logical rules}

  ("\RR" of the form ". ; Q @ w" or "N @ w ; .", and "\LL" of the form
  "\G ; \D ; \W")
  \begin{gather*}
    \I["tens L"]{"\LL, P tens Q @ u ==> \RR"}{"\LL, P @ u, Q @ u ==> \RR"}
    \quad
    \I["one L"]{"\LL, one @ u ==> \RR"}{"\LL ==> \RR"}
    \quad
    \I["plus L"]{"\LL, P plus Q @ u ==> \RR"}{
      "\LL, P @ u ==> \RR" & "\LL, Q @ u ==> \RR"
    }
    \quad
    \I["zero L"]{"\LL, zero @ u ==> \RR"}
    \\[1ex]
    \I["{dn}LA"]{"\LL, now u. P @ v ==> \RR"}{"\LL, [v/u] P @ v ==> \RR"}
    \quad
    \I["{at}LA"]{"\LL, (P at u) @ v ==> \RR"}{"\LL, P @ u ==> \RR"}
    \quad
    \I["\exists L^\alpha"]{"\LL, \fex \alpha P @ u ==> \RR"}{"\LL, P @ u ==> \RR"}
    \\[1ex]
    \I["{!L}"]{"\G ; \D ; \W, {!N @ u} ==> \RR"}{"\G, N @ u ; \D ; \W ==> \RR"}
    \quad
    \I["pos L"]{"\G ; \D ; \W, {pos N @ w} ==> \RR"}{"\G ; \D, N @ w ; \W ==> \RR"}
    \quad
    \I[lp]{"\G ; \D ; \W, {p\ \vec t} @ w ==> \RR"}{"\G ; \D, neg p\ \vec t ; \W ==> \RR"}
    \\[1ex]
    \I["with R"]{"\LL ==> M with N @ w ; ."}{"\LL ==> M @ w ; ." & "\LL ==> N @ w ; ."}
    \quad
    \I["top R"]{"\LL ==> top @ w ; ."}
    \quad
    \I["{-o}R"]{"\LL ==> P -o N @ w ; ."}{"\LL, P @ w ==> N @ w ; ."}
    \\[1ex]
    \I["{dn}RA"]{"\LL ==> now u. N @ w ; ."}{"\LL ==> [w/u] N @ w ; ."}
    \quad
    \I["{at}RA"]{"\LL ==> (N at u) @ w"}{"\LL ==> N @ u"}
    \quad
    \I["\forall R^\alpha"]{"\LL ==> \fall \alpha N @ u ; ."}{"\LL ==> N @ u ; ."}
    \\[1ex]
    \I["neg R"]{"\LL ==> {neg P @ w} ; ."}{"\LL ==> . ; P @ w"}
    \quad
    \I[rp]{"\LL ==> n\ \vec t @ w ; ."}{"\LL ==> . ; pos n\ \vec t @ w"}
  \end{gather*}

  \paragraph{Focusing decisions}

  ("\LL" of the form "\G ; \D")
  \begin{gather*}
    \I[lf]{"\G ; \D, N @ u ; . ==> . ; Q @ w"}{"\G ; \D ; foc {N @ u} ==> Q @ w" & N \text{ not } "neg p\ \vec t"}
    \quad
    \I[cplf]{"\G, N @ u ; \D ; . ==> . ; Q @ w"}{"\G, N @ u ; \D ; foc {N @ u} ==> Q @ w"}
    \\[1ex]
    \I[rf]{"\G ; \D ; . ==> . ; P @ w"}{"\G ; \D ==> foc {P @ w}" & P \text{ not } "pos n\ \vec t"}
  \end{gather*}
  \egroup
  \end{minipage}}
\caption{Focusing rules for \hyll.}
\label{fig:foc-rules}
\end{figure}

\medskip
\noindent
The two syntactic classes refer to each other via the new connectives "neg" and
"pos". Sequents in the focusing calculus are of the following forms.
\begin{center} \small
  \begin{tabular}{r@{\ }l@{\qquad}r@{\ }l}
    \( \left.
    \begin{array}[c]{l}
      "\G ; \D ; \W ==> . ; P @ w" \\
      "\G ; \D ; \W ==> N @ w ; ."
    \end{array}
    \right\} \) & active &
    \( \left.
    \begin{array}[c]{l}
      "\G ; \D ; foc{N @ u} ==> P @ w" \\
      "\G ; \D ==> foc{P @ w}" \\
    \end{array}
    \right\} \) & focused \\
  \end{tabular}
\end{center}
In each case, "\G" and "\D" contain only negative propositions (\ie, of the form
"N @ u") and "\W" only positive propositions (\ie, of the form "P @ u"). 
The full collection of inference rules are in \figref{foc-rules}.
The sequent form "\G ; \D ; . ==> . ; P @ w" is called a \emph{neutral sequent}; from
such a sequent, a left or right focused sequent is produced with the rules lf,
cplf or rf. Focused logical rules are applied (non-deterministically) and focus
persists unto the subformulas of the focused proposition as long as they are of
the same polarity; when the polarity switches, the result is an active sequent,
where the propositions in the innermost zones are decomposed in an irrelevant
order until once again a neutral sequent results.

Soundness of the focusing calculus with respect to the ordinary sequent calculus
is immediate by simple structural induction. In each case, if we forget the
additional structure in the focused derivations, then we obtain simply an
unfocused proof. We omit the obvious theorem. Completeness, on the other hand,
is a hard result. We omit the proof because focusing is by now well known for
linear logic, with a number of distinct proofs via focused cut-elimination (see
\eg the detailed proof in~\cite{chaudhuri08jar}). The hybrid connectives pose no
problems because they allow all cut-permutations.

\bgroup 
\begin{thm}[focusing completeness]
  Let "\G^-" and "C^- @ w" be negative polarizations of "\G" and "C @ w" (that
  is, adding "neg" and "pos" to make "C" and each proposition in "\G" negative)
  and "\D^+" be a positive polarization of "\D". If "\G ; \D ==> C @ w", then
  ". ; . ; {!  \G^-}, \D^+ ==> C^- @ w ; .".
\end{thm}
\egroup

\section{Encoding the synchronous stochastic $\pi$-calculus}
\label{sec:spi}

In this section, we shall illustrate the use of \hyllp as a logical framework
for constrained transition systems by encoding the syntax and the operational
semantics of the synchronous stochastic $\pi$-calculus (\spi), which extends the
ordinary $\pi$-calculus by assigning to every channel and internal action
an \emph{inherent} rate of synchronization.
In \spi, each rate characterises an exponential distribution\cite{phillips06tcsb}, 
such that the probability of a reaction with rate $r$ 
is given by
$\text{Prob}(X_t \le x + rt \mid X_0=x) = {\mathbf F}_{X_t \mid X_0=x}(x + rt) = 1 - e^{rt}$,
where the {\emph rates} $r$ are functions depending on the time $t$.
We have seen in \secref{hyllp} that, in this case, we can use 
a particular instance of \hyllp,
where the worlds $w$ are $\mu_X(x+rt)$, 
representing the probability of $X$ to have its value less or equal than $x+rt$.
Note that the distributions have the same shape for any variables $X_t$; 
They only depend on {\emph rates} $r(t)$ and time $t$.
We shall use this fact to encode \spi in \hyllp: 
a \spi reaction with rate $r$ will be encoded by a transition 
of probability 
$w = \mu_X(x + rt) = 
\text{Prob}(X_t \le x + rt \mid X_0=x) = {\mathbf F}_{X_t \mid X_0=x}(x + rt) = 1 - e^{rt}$.
In the rest of this section, worlds $w$ of this shape,
defined by a rate $r$, will simply be written $r$ 
(see \defnref{sencoding} and \defnref{iencoding}).

\hyllp can therefore be seen as a formal language for expressing \spi executions (traces). 
For the rest of this
section we shall use "r, s, t, \ldots" instead of "u, v, w, \ldots" to highlight
the fact that the worlds represent (probabilities defined by) rates, 
with the understanding that "." is convolution (fact \ref{fact:convolution}) 
and "rid" is ${\mathop{\Theta}}$. We don't directly use
rates because the syntax and transitions of \spi are given generically for a
$\pi$-calculus with labelled actions, and it is only the interpretation of the
labels that involves probabilities.

We first summarize the
syntax of \spi, which is a minor variant of a number of similar presentations
such as~\cite{phillips06tcsb}. For hygienic reasons we divide entities into the
syntactic categories of \emph{processes} ($\pr P, \pr Q, \ldots$) and
\emph{sums} ($\pr M, \pr N, \ldots$), defined as follows. We also include
environments of recursive definitions ("pr E") for constants.

\smallskip
\bgroup 
\begin{tabular}{l@{\quad}l@{\ }r@{\ }l}
\emph{(Processes)} & "pr P, pr Q, ..." & "::=" & "pr{\nu_r\ P} OR pr{P par Q} OR pr 0 OR pr{X_n\,x_1 \cdots x_n} OR pr M" \\
\emph{(Sums)} & "pr M, pr N, ..." & "::=" & "pr {act{oup{x}(y)} P} OR pr {act{inp{x}} P} OR pr{act{\tau_r} P} OR pr{M + N}" \\
\emph{(Environments)} & "pr E" & "::=" & "pr {E, X_n \triangleq P} OR pr ."
\end{tabular}
\egroup

\smallskip

"pr{P par Q}" is the parallel composition of "pr P" and "pr Q", with unit "pr
0". The restriction "pr{\nu_r\ P}" abstracts over a free channel "x" in the
process "pr{P\,x}". We write the process using higher-order abstract
syntax~\cite{pfenning88pldi}, \ie, "pr{P}" in "pr{\nu_r\ P}" is (syntactically)
a function from channels to processes. This style lets us avoid cumbersome
binding rules in the interactions because we reuse the well-understood binding
structure of the $\lambda$-calculus. A similar approach was taken in the
earliest encoding of (ordinary) $\pi$-calculus in (unfocused) linear
logic~\cite{miller92welp}, and is also present in the encoding in
CLF~\cite{cervesato03tr}.

A sum is a non-empty choice ("+") over terms with \emph{action prefixes}: the
output action "oup{x}(y)" sends "y" along channel "x", the input action "inp{x}"
reads a value from "x" (which is applied to its continuation process), and the
internal action "\tau_r" has no observable I/O behaviour. Replication of
processes happens via guarded recursive definitions~\cite{milner99book};
in~\cite{Regev01psb} it is argued that they are more practical for programming
than the replication operator "!". In a definition "pr{X_n \triangleq P}", "pr
{X_n}" denotes a (higher-order) defined constant of arity "n"; given channels
"x_1, ..., x_n", the process "pr {X_n\,x_1 \cdots x_n}" is synonymous with
"pr{P\,x_1 \cdots x_n}". The constant "pr{X_n}" may occur on the right hand side
of any definition in "pr E", including in its body "pr P", as long as it is
prefixed by an action; this prevents infinite recursion without progress.

Interactions are of the form "pr E |- pr P ->[r] pr Q" denoting a transition
from the process "pr P" to the process "pr Q", in a global environment "pr E",
by performing an action at rate "r". Each channel "x" is associated with an
inherent rate specific to the channel, and internal actions "\tau_r" have rate
"r". The restriction "pr{\nu_r\ P}" defines the rate of the abstracted channel
as "r".

\sdef{crate}{\mathop{\mathrm{rate}}}

\begin{figure*}[tp]
\small
\hspace{-1.8em}
\framebox{ \begin{minipage}{1.05\textwidth}
  \emph{Interactions} \vspace{-1em}

  \begin{gather*}
    \I[\set{SYN}]{"pr{act{oup{x}(y)} P + M par act{inp{x}} Q + M'} ->[crate(x)] pr{P par Q\,y}"}{}
    \SP
    \I[\set{INT}]{"pr{act{\tau_r} P} ->[r] pr P"}{}
    \\
    \I[\set{PAR}]{"pr{P par Q} ->[r] pr{P' par Q}"}{"pr P ->[r] pr{P'}"}
    \SP
    \I[\set{RES}]{"pr{\nu_s\ P} ->[r] pr{\nu_s\ Q}"}{
      \forall x_s. \Bigl("pr{P\,x} ->[r] pr{Q\,x}"\Bigr)
    }
    \SP
    \I[\set{CONG}]{"pr{P'} ->[r] pr{Q'}"}{"pr{P} ->[r] pr{Q}" & "pr{P} == pr{P'}" & "pr{Q} == pr{Q'}"}
  \end{gather*}

  \vspace{-1.5em} \dotfill

  \emph{Congruence} \vspace{-1em}
  \begin{gather*}
    \I{"pr{P par 0} == pr P"}{}
    \LSP
    \I{"pr{P par Q} == pr{Q par P}"}{}
    \LSP
    \I{"pr{P par (Q par R)} == pr{(P par Q) par R}"}{}
    \LSP
    \I{"pr{\nu_r\ 0 == 0}"}{}
    \LSP
    \I{"pr{E} |- pr{X_n\,x_1 \cdots x_n} == pr{P\,x_1 \cdots x_n}"}
      {"pr{X_n \triangleq P} \in pr{E} "}
    \\
    \I{"pr{\nu_r (lam x. \nu_s (lam y. P))} == pr{\nu_s (lam y. \nu_r (lam x. P))}"}{}
    \LSP
    \I{"pr{\nu_r\ P} == pr{\nu_r\ Q}"}{
      \forall x_r.\left("pr{P\,x} == pr{Q\,x}"\right)
    }
    \LSP
    \I{"pr{\nu_r (lam x. P par Q(x))} == pr{P par \nu_r\ Q}"}
    \\
    \I{"pr{P par Q} == pr{P' par Q}"}{"pr P == pr{P'}"}
    \LSP
    \I{"pr{act{oup{x}(m)} P} == pr{act{oup{x}(m)} P'}"}{"pr P == pr {P'}"}
    \LSP
    \I{"pr{act{inp{x}} P} == pr{act{inp{x}} Q}"}{
      \forall n.\ \left("pr{P\,n} == pr{Q\,n}"\right)
    }
    \LSP
    \I{"pr{act{\tau_r} P} == pr{act{\tau_r} P'}"}{
      "pr P == pr {P'}"
    }
    \\
    \I{"pr{M + N} == pr{N + M}"}
    \LSP
    \I{"pr{M + (N + K)} == pr{(M + N) + K}"}
    \LSP
    \I{"pr{M + N} == pr{M' + N}"}
      {"pr M == pr {M'}"}
    \LSP
    \I{"pr{M + N} == pr M"}{"pr M == pr N"}
  \end{gather*}

\end{minipage}}
\caption{Interactions and congruence in \spi. The environment $E$ is elided in most rules.}
\label{fig:spi}
\end{figure*}

The full set of interactions and congruences are in fig.~\ref{fig:spi}. We
generally omit the global environment "pr E" in the rules as it never changes.
It is possible to use the congruences to compute a normal form for processes
that are a parallel composition of sums and each reaction selects two suitable
sums to synchronise on a channel until there are no further reactions possible;
this refinement of the operational semantics is used in "spi" simulators such as
SPiM~\cite{phillips04bc}.

\bgroup 
\begin{defn}[syntax encoding] \mbox{} \label{defn:sencoding}
  \begin{ecom}[1.]
  \item The encoding of the process "pr P" as a positive proposition, written
    "proc P", is as follows ("dt" is a positive atom and \crt a negative atom).
    \begin{align*}
      "proc{P par Q}" &= "proc P tens proc Q"
      & \SP
      "proc {\nu_r\ P}" &= "ex x. ! (rt x at r) tens proc {P\,x}"
      \\
      "proc 0" &= "one"
      &
      "proc{X_n\,x_1 \cdots x_n}" &= "X_n\,x_1 \cdots x_n"
      \\
      "proc M" &= "pos {(dt -o sum M)}"
    \end{align*}
  \item The encoding of the sum "pr M" as a negative proposition, written "sum
    M", is as follows (\cn{out}, \cn{in} and \cn{tau} are positive atoms).
    \begin{align*}
      "sum {M + N}" &= "sum M with sum N"
      &\hspace{-2em}
      "sum {act{oup x(m)} P}" &= "neg (\cout x~m tens proc P)"
      \\
      "sum {act{inp x} P}" &= "all n. neg (\cin x~n tens proc {P\,n})"
      &
      "sum {act{\tau_r} P}" &= "neg (\ctau r tens proc P)"
    \end{align*}
  \item The encoding of the definitions "pr E" as a context, written "env E", is
    as follows.
    \begin{align*}
      "env{E, X_n \triangleq P} " \ & = \ 
         "env{E}, !! \fall{x_1, ..., x_n} X_n\,x_1 \cdots x_n o-o proc{P\,x_1 \cdots x_n} " \\
      "env{ . } " \ & = \  " . " 
    \end{align*}
    where "P o-o Q" is defined as "(P -o neg Q) with (Q -o neg P)".
  \end{ecom}
\end{defn}
\egroup

The encoding of processes is positive, so they will be decomposed in the active
phase when they occur on the left of the sequent arrow, leaving a collection of
sums.  The encoding of restrictions will introduce a fresh unrestricted
assumption about the rate of the restricted channel. Each sum encoded as a
processes undergoes a polarity switch because "-o" is negative; the antecedent
of this implication is a \emph{guard} "dt". This pattern of guarded switching of
polarities prevents unsound congruences such as "pr{act{oup x(m)} act{oup y(n)}
  P} == pr{act{oup y(n)} act{oup x(m)} P}" that do not hold for the synchronous
$\pi$ calculus.\footnote{Note: $"(x -o a tens (x -o b tens c))" \multimap "(x -o
  b tens (x -o a tens c))"$ is not provable in linear logic.}  This guard also
\emph{locks} the sums in the context: the \spi interaction rules \set{INT} and
\set{SYN} discard the non-interacting terms of the sum, so the environment will
contain the requisite number of "dt"s only when an interaction is in progress.
The action prefixes themselves are also synchronous, which causes another
polarity switch. Each action releases a token of its respective kind---\cn{out},
\cn{in} or \cn{tau}---into the context. 
These tokens must be consumed by the interaction before 
the next interaction occurs.
For each action, 
the (encoding of the) continuation process is also released into the context.

The proof of the following congruence lemma is omitted. Because the encoding is
(essentially) a "tens / with" structure, there are no distributive laws in
linear logic that would break the process/sum structure.

\begin{thm}[congruence] 
  \label{thm:congr} \hfill\break
  "pr E |- pr P == pr Q" iff both "env E @ rid ; . ; proc P @ rid ==> . ; proc Q @ rid" 
  and "env E @ rid ; . ; proc Q @ rid ==> . ; proc P @ rid".
\end{thm}

Now we encode the interactions. Because processes were lifted into propositions,
we can be parsimonious with our encoding of interactions by limiting ourselves
to the atomic interactions \set{syn} and \set{int} (below); the \set{par},
\set{res} and \set{cong} interactions will be ambiently implemented by the
logic. Because there are no concurrent interactions---only one interaction can
trigger at a time in a trace---the interaction rules must obey a locking
discipline. We represent this lock as the proposition \cact that is consumed at
the start of an interaction and produced again at the end. This lock also
carries the net rate of the prefix of the trace so far: that is, an interaction
"pr P ->[r] pr Q" will update the lock from "\cact @ s" to "\cact @ {s .
  r}". The encoding of individual atomic interactions must also remove the
\cn{in}, \cn{out} and \cn{tau} tokens introduced in context by the interacting
processes.

\begin{defn}[interaction] \mbox{} \newline
\label{defn:iencoding}
  Let "\cinter \triangleq !! (\cact -o neg \cint with neg \csyn)" where \cact is
  a positive atom and \cint and \csyn are as follows:
  \begin{align*}
    \cint &\triangleq "(dt at rid) tens pos all r. \Bigl((\ctau r at rid) -o rate r neg \cact\Bigr)" \\
    \csyn &\triangleq "(dt tens dt at rid) tens pos all x, r, m. \Bigl((\cout x~m tens \cin x~m at rid) -o pos (rt x at r) -o rate r neg \cact \Bigr)".
   \end{align*}
\end{defn}

\noindent
The number of interactions that are allowed depend on the number of instances of
\cinter in the linear context: each focus on \cinter implements a single
interaction. If we are interested in all finite traces, we will add \cinter to
the unrestricted context so it may be reused as many times as needed.

\subsection{Representational adequacy.}
\label{sec:spi.adq}

Adequacy consists of two components: completeness and soundness. Completeness is
the property that every \spi execution is obtainable as a \hyll derivation using
this encoding, and is the comparatively simpler direction (see
\thmref{completeness}). Soundness is the reverse property, and is false for
unfocused \hyll as such. However, it \emph{does} hold for focused proofs (see
\thmref{adeq}). In both cases, we reason about the following canonical sequents
of \hyll.

\begin{defn} 
  The \emph{canonical context of} "P", written "can P", is given by:
  \begin{align*} 
    "can {X_n\,x_1 \cdots x_n}" &= "neg {X_n\,x_1 \cdots x_n}" &
    "can {P par Q}" &= "can P", "can Q" &
    "can 0" &= "." &
    "can {\nu_r\ P}" &= "can{P\,a}" \\
    "can M" &= "dt -o sum M" &
  \end{align*}
  For "can{\nu_r\ P}", the right hand side uses a \emph{fresh} channel "a" that
  is not free in the rest of the sequent it occurs in.
\end{defn}

\noindent
As an illustration, take "pr P \triangleq pr{act{oup x(a)} Q par act{inp x} R}". We have:
\begin{gather*} 
  "can P" =
  \begin{array}[t]{l}
    "dt -o neg (\cout x~a tens proc Q)", 
    "dt -o all y. neg (\cin x~y tens  proc {R\,y})"
  \end{array}
\end{gather*}
Obviously, the canonical context is what would be emitted to the linear zone at
the end of the active phase if "proc P" were to be present in the left active
zone.

\begin{defn} 
  A neutral sequent is \emph{canonical} iff it has the shape
  \begin{gather*} 
    "env{E}, \cn{rates}, \cinter @ rid ;
     neg \cact @ s, can {P_1 par \cdots par P_k} @ rid ;
     . ==> . ; (proc Q at rid) tens \cact @ t"
  \end{gather*}
  where \cn{rates} contains elements of the form "rt x @ r" defining the
  rate of the channel "x" as "r", and all free channels in "env{E}, can{P_1 par \cdots
    par P_k par Q}" have a single such entry in \cn{rates}.
\end{defn}

\begin{figure*}[tp]
  \centering
  \small
  \begin{multline*}
    \text{Suppose "\LL = \crt x @ r, \cinter @ rid" and "\RR = (proc S
      at rid) tens \cact @ t". (All judgements "{} @ rid" omitted.)}
    \\[1ex]
    \infer={"\LL ; neg \cact @ s, can{act{oup x(a)} Q par act{inp x} R} ; . ==> . ; \RR"}{
    \I[1]{"\LL ;
              neg \cact @ s,
              dt -o neg (\cout x\ a tens proc Q),
              dt -o all y. neg (\cin x~y tens proc {R\,y})
              ; . ==> . ; \RR"}
         {\I[2]{"\LL ; 
                     neg \cact @ s,
                     dt -o neg (\cout x\ a tens proc Q),
                     dt -o all y. neg (\cin x~y tens proc {R\,y})
                     ; foc{\cinter} ==> \RR"
               }
               {\I[3]{\LL\ ;\ 
                      \begin{array}[t]{l}
                        "dt -o neg (\cout x\ a tens proc Q)",
                        "dt -o all y. (\cin x~y tens proc {R\,y})", \\
                        "neg dt, neg dt, all x, r, m. ((\cout x~m tens \cin x~m at rid) -o pos(\crt x at r) -o rate {r} \cact) @ s"
                      \end{array}
                      \begin{array}[t]{l}
                        {} \\ "{} ; . ==> . ; \RR"
                      \end{array}
                     }
                     {\I[4]{\LL\ ; \
                            \begin{array}[t]{l}
                              "neg \cout x\ a", "can Q",
                              "dt -o all y. neg (\cin x~y tens proc {R\,y})", \\
                              "neg dt, all x, r, m. ((\cout x~m tens \cin x~m at rid) -o pos(\crt x at r) -o rate {r} \cact) @ s"
                            \end{array}
                            \begin{array}[t]{l}
                              {} \\ "{} ; . ==> . ; \RR"
                            \end{array}
                           }
                           {\I[5]{\LL\ ;\
                                  \begin{array}[t]{l}
                                    "can Q", "neg \cout x~a", "neg \cin x~a", "can{R\,a}", \\
                                    "all x, r, m. ((\cout x~m tens \cin x~m at rid) -o pos(\crt x at r) -o rate {r} \cact) @ s"
                                  \end{array}
                                  \begin{array}[t]{l}
                                    {} \\ "{} ; . ==> . ; \RR"
                                  \end{array}
                                 }
                                 {"\LL ; can{Q}, can{R\,a}, neg \cact @ {s . r} ; . ==> . ; \RR"}
                           }
                     }
               }
         }
       }\\[-2em]
  \end{multline*}
  \begin{tabular}{l@{:\ }l@{\SP}l@{:\ }l@{\SP}l@{:\ }l}
    \multicolumn{2}{l}{Steps}\\\hline
    1 & focus on "\cinter \in \LL" &
    3 & "dt" for output + full phases &
    5 & cleanup \\
    2 & select \csyn from \cinter, active rules & 
    4 & "dt" for input + full phases
  \end{tabular}
  \caption{Example interaction in the \spi-encoding.}
  \label{fig:example-syn}
\end{figure*}

\figref[Figure]{example-syn} contains an example of a derivation for a canonical
sequent involving "pr P". Focusing on any (encoding of a) sum in "can P @ rid"
will fail because there is no "dt" in the context, so only \cinter can be given
focus; this will consume the \cact and release two copies of "(dt at rid)" and
the continuation into the context. Focusing on the latter will fail now (because
"\cout x~m" and "\cin x~m" (for some "m") are not yet available), so the only
applicable foci are the two sums that can now be ``unlocked'' using the
"dt"s. The output and input can be unlocked in an irrelevant order, producing
two tokens "\cin x~a" and "\cout x~a". Note in particular that the witness "a"
was chosen for the universal quantifier in the encoding of "pr{act{inp{x}}Q}"
because the subsequent consumption of these two tokens requires the messages to
be identical. (Any other choice will not lead to a successful proof.)  After
both tokens are consumed, we get the final form "\cact @ {s . r}", where "r" is
the inherent rate of "x" (found from the \cn{rates} component of the
unrestricted zone). This sequent is canonical and contains "can{Q par R\,a}".

Our encoding therefore represents every \spi action in terms of ``micro''
actions in the following rigid order: one micro action to determine what kind of
action (internal or synchronization), one micro action per sum to select the
term(s) that will interact, and finally one micro action to establish the
contract of the action. Thus we see that focusing is crucial to maintain the
semantic interpretation of (neutral) sequents. In an unfocused calculus, several
of these steps could have partial overlaps, making such a semantic
interpretation inordinately complicated. We do not know of any encoding of the
$\pi$ calculus that can provide such interpretations in unfocused sequents
without changing the underlying logic. In CLF~\cite{cervesato03tr} the logic is
extended with explicit monadic staging, and this enables a form of
adequacy~\cite{cervesato03tr}; however, the encoding is considerably more
complex because processes and sums cannot be fully lifted and must instead be
specified in terms of a lifting computation. Adequacy is then obtained via a
permutative equivalence over the lifting operation. Other encodings of $\pi$
calculi in linear logic, such as~\cite{garg05concur} and~\cite{baelde05stage},
concentrate on the easier asynchronous fragment and lack adequacy proofs anyhow.

\def\sproc{\set{proc}}
\def\ssum{\set{sum}}
\def\sinter{\set{inter}}

\begin{thm}[completeness] \label{thm:completeness} 
  If "pr E |- pr P ->[r] pr Q", then the following canonical sequent is derivable.
  \begin{gather*}
    "env{E}, \cn{rates}, \cinter @ rid ; neg \cact @ s, can P @ rid ; . ==> . ; (proc Q at rid) tens \cact @ {s . r}".
  \end{gather*}
\end{thm}

\begin{proof}
  By structural induction of the derivation of "pr E |- pr P ->[r] pr Q". Every
  interaction rule of \spi is implementable as an admissible inference rule for
  canonical sequents. For \set{cong}, we appeal to \thmref{congr}.
\end{proof}

Completeness is a testament to the expressivity of the logic -- all executions
of \spi are also expressible in \hyll. However, we also require the opposite
(soundness) direction: that every canonical sequent encodes a possible \spi
trace. The proof hinges on the following canonicity lemma.

\begin{lem}[canonical derivations] \label{lem:canonical} 
  In a derivation for a canonical sequent, the derived inference rules for
  \cinter are of one of the two following forms (conclusions and premises
  canonical).
  \begin{gather*}
    \I{"env{E}, \cn{rates}, \cinter @ rid ; neg \cact @ s, can P @ rid ; . ==> . ; (proc P at rid) tens \cact @ s"}{}
    \\[1ex]
    \I{"env{E}, \cn{rates}, \cinter @ rid ; neg \cact @ s, can P @ rid ; . ==> . ; (proc R at rid) tens \cact @ t"}
      {"env{E}, \cn{rates}, \cinter @ rid ; neg \cact @ {s . r}, can Q @ rid ; . ==> . ; (proc R at rid) tens \cact @ t"}
  \end{gather*}
  where: either "pr E |- pr P ->[r] pr Q", or "pr E |- pr P == pr Q" with "r = rid".
\end{lem}

\begin{proof}
  This is a formal statement of the phenomenon observed earlier in the example
  (\figref{example-syn}): "proc R tens \cact" cannot be focused on the right
  unless "pr P == pr R", in which case the derivation ends with no more foci on
  \cinter. If not, the only elements available for focus are \cinter and one of
  the congruence rules "env{E}" in the unrestricted context. In the former case,
  the derived rule consumes the "neg \cact @ s", and by the time \cact is
  produced again, its world has advanced to "s . r". In the latter case, the
  definition of a top level "X_n" in "can P" is (un)folded (without advancing
  the world). 
  The proof proceeds by induction on the structure of $P$.
\end{proof}

\lemref[Lemma]{canonical} is a strong statement about \hyll derivations using
this encoding: every partial derivation using the derived inference rules
represents a prefix of an \spi trace. This is sometimes referred to as
\emph{full adequacy}, to distinguish it from adequacy proofs that require
complete derivations~\cite{nigam08ijcar}. The structure of focused derivations
is crucial because it allows us to close branches early (using init). It is
impossible to perform a similar analysis on unfocused proofs for this encoding;
both the encoding and the framework will need further features to implement a
form of staging~\cite[Chapter 3]{cervesato03tr}.

\begin{cor}[soundness] \label{thm:adeq}  \mbox{} \newline
  If "env{E}, \cn{rates}, \cinter @ rid ; neg \cact @ rid, can P @ rid
  ; .  ==> . ; (proc Q at rid) tens \cact @ r" is derivable, then "pr E |- pr
  P ->[r]{\!}^* pr Q".
\end{cor}

\begin{proof}
  Directly from \lemref{canonical}.
\end{proof}

\subsection{Stochastic correctness with respect to simulation}
\label{sec:simul}

So far the \hyllp encoding of \spi represents any \spi trace symbolically. 
However, not every symbolic trace of an \spi process can be produced according to the
operational semantics of \spi, which is traditionally given by a simulator.
This is the main
difference between \hyll (and \spi) and the approach of CSL~\cite{aziz00tcl},
where the truth of a proposition is evaluated against a CTMC, which is why
equivalence in CSL is identical to CTMC bisimulation~\cite{desharmais03jlap}. In
this section we sketch how the execution could be used directly on the canonical
sequents to produce only correct traces (proofs). The proposal in this section
should be seen by analogy to the execution model of \spi simulators such as
SPiM~\cite{phillips04cmmb}, although we do not use the Gillespie algorithm.

The main problem of simulation is determining which of several competing enabled
actions in a canonical sequent to select as the ``next'' action from the
\emph{race condition} of the actions enabled in the sequent. Because of the
focusing restriction, these enabled actions are easy to compute. Each element of
"can P" is of the form "dt -o sum M", so the enabled actions in that element are
given precisely by the topmost occurrences of "neg" in "sum M". Because none of
the sums can have any restricted channels (they have all been removed in the
active decomposition of the process earlier), the rates of all the channels will
be found in the "\cn{rates}" component of the canonical sequent. 

The effective rate of a channel "x" is related to its inherent rate by scaling
by a factor proportional to the \emph{activity} on the channel, as defined
in~\cite{phillips04cmmb}. Note that this definition is on the \emph{rate
  constants} of exponential distributions, not the rates themselves. The
distribution of the minimum of a list of random variables with exponential
distribution is itself an exponential distribution whose rate constant is the
sum of those of the individual variables. Each individual transition on a
channel is then weighted by the contribution of its rate to this sum. The choice
of the transition to select is just the ordinary logical non-determinism. Note
that the rounds of the algorithm do not have an associated \emph{delay} element
as in~\cite{phillips04cmmb}; instead, we compute (symbolically) a distribution
over the delays of a sequence of actions.

Because stochastic correctness is not necessary for the main adequacy result in
the previous subsection, we leave the details of simulation to future work.

\section{Direct encoding of molecular biology}
\label{sec:bio}

Models of molecular biology have a wealth of examples of transition systems with
temporal and stochastic constraints.  In a biochemical reaction, molecules can
interact to form other molecules or undergo internal changes such as
phosphorylation, and these changes usually occur as parts of networks of
interacting processes with continuous kinetic feedback. \spi has been used in a
number of such models; since we have an adequate encoding of \spi, we can use
these models via the encoding.

However, biological systems can also be encoded directly in \hyll. As an
example, consider a simplified \emph{repressilator} gene network consisting of
two genes, each causing the production of a protein that represses the other
gene by negative feedback. This is a simplification of the three-gene network
constructed in~\cite{elowitz00nature}.  We note that each gene can be in an
``on'' (activated) or an ``off'' (deactivated) state, represented by the unary
predicates \cn{on} and \cn{off}. Molecules of the transcribed proteins are
represented with the unary predicate \cn{prot}. Transitions in the network are
encoded as axioms.

\paragraph{Example: the repressilator, using temporal constraints}

The system consists of the following components:
\begin{icom}
\item \emph{Repression}: Each protein molecule deactivates the next gene in
  the cycle after (average) deactivation delay "d"
  \begin{quote}
    "cn{repress}\,a\,b \ def \ cn{prot}\,a tens cn{on}\,b -o rate d (cn{off}\,b tens cn{prot}\,a)".
  \end{quote}

\item \emph{Reactivation}: When a gene is in the ``off'' state, it eventually
  becomes ``on'' after an average delay of "r":
  \begin{quote}
    "cn{react} \ def \ all a. cn{off}\,a -o rate r cn{on}\,a".
  \end{quote}
  It is precisely this reactivation that causes the system to oscillate instead
  of being bistable.

\item \emph{Synthesis}: When a gene is ``on'', it transcribes RNA for its
  protein taking average delay "t", after which it continues to be ``on'' and a
  molecule of the protein is formed.
  \begin{quote}
    "cn{synt} \ def \ all a. cn{on}\,a -o rate t (cn{on}\,a tens cn{prot}\,a)".
  \end{quote}

\item \emph{Dissipation}: If a protein does not react with a gene, then it
  dissipates after average delay "s":
  \begin{quote}
    "cn{diss} \ def \ all a. cn{prot}\,a -o rate s one".
  \end{quote}

\item \emph{Well defined}: We need to say that a gene cannot be \textrm{on} and 
     \textrm{off} at the same time, that a gene has to be \textrm{on} or \textrm{off}, 
     and that all delays are different:

  \qquad $\cn{well\_def}~ \eqdef~ 
           (\forall x.~ \textrm{on} \, x \otimes \textrm{off} \, x \limp 0 ) \land
           (\forall x.~ \textrm{on} \, x \lor \textrm{off} \, x) \land
           (d \neq r \neq s \neq t) $.
\end{icom}

\noindent
The system consists of a repression cycle for genes \cn{a} and \cn{b}, and the other processes:
\begin{quote}
  "cn{system} \ def \ cn{repress}\,cn{a}\,cn{b}, cn{repress}\,cn{b}\,cn{a}, cn{react}, cn{synt}, cn{diss}," \cn{well\_def}.
\end{quote}

\noindent
Examples of valid sequents are ("0" is the initial instant of time):
\begin{gather*}
  "!! cn{system} @ 0"
  \ ; \ 
  \underbrace{"rate{r + t}cn{on}\,{cn a} @ 0, cn{off}\,{cn b} @ 0"}_{\text{initial state}}
  \ "==>" 
  \underbrace{"rate{r + t + d}cn{off}\,{cn a} tens top @ 0"}_{\text{final state}}
\end{gather*}
From "cn{off}\,cn{b}" we get "cn{on}\,cn{b} tens cn{prot}\,cn{b}" after interval "r + t"; then
"cn{prot}\,cn{b}" together with "cn{on}\,cn{a}" forms "cn{prot}\,cn{b} tens cn{off}\,cn{a}" after a
further delay "d".

Note. In general, delays are functions of the elements involved. Handling this is feasable; However it would require to extend HyLL syntax.

\paragraph{Example: stochastic repressilator}

We now revisit our example but this time using rates.  Note that the encodings
can be very similar in the temporal and probabilistic fragments of our logic; the
only differences being the interpretation of the constraints: Here, $d, t, r$
and $s$ are interpreted as (probabilities defined by) rates.
\begin{align*}
  "cn{repress}~a~b" &\DEF "cn{prot}~a tens cn{on}~b -o rate{d} (cn{off}~b tens cn{prot}~a)" \\
  "cn{synt}" &\DEF "all a. cn{on}~a -o rate{t} (cn{on}~a tens cn{prot}~a)". \\
  "cn{react}" &\DEF "all a. cn{off}~a -o rate{r} cn{on}~a". \\
  "cn{diss}" &\DEF "all a. cn{prot}~a -o rate{s} one". 
\end{align*}
Suppose we want to show that in the two-gene repressilator, the state "cn{on}(a)
tens cn{off}(b)" can oscillate to "cn{off}(a) tens cn{on}(b)". The proof looks
as below, with one sub-proofs named "P", and most of the worlds and a
second sub-proof elided:
\begin{gather*}
  \I[\cn{react}]{"cn{on}~cn{a}, cn{off}~cn{b} ==> ex k. rate{k} (cn{off}~cn{a} tens cn{on}~cn{b})"}
    {\I{"cn{off}~cn{b} ==> cn{off}~cn{b}"}{}
     &
     \I[\cn{synt}]{"cn{on}~cn{a}, rate{r} cn{on}~cn{b} ==> ex k. rate{k} (cn{off}~cn{a} tens cn{on}~cn{b})"}
       {\I{"cn{on}~cn{b} ==> cn{on}~cn{b}"}
        &
        \I{"cn{on}~cn{a}, rate{r} rate{t} (cn{on}~cn{b} tens cn{prot}~cn{b}) ==> ex k. rate{k} (cn{off}~cn{a} tens cn{on}~cn{b})"}
          {\I[$P$]{"cn{on}~cn{a}, rate{r} rate{t} cn{prot}~cn{b} ==> rate{r} rate{t} rate{d} cn{off}~cn{a}"}
           &
           \cdots
          }
       }
    }
  \\[3ex]
  \I[$P\ =$]{} \
  \I[\cn{repress}~\cn{b}~\cn{a}]{"cn{on}~cn{a}, rate{r} rate{t} cn{prot}~cn{b} ==> rate{r} rate{t} rate{d} cn{off}~cn{a}"}
    {\I["tens I"]{"cn{on}~cn{a}, rate{r} rate{t} cn{prot}~cn{b} ==> rate{r} rate{t} (cn{on}~cn{a} tens cn{prot}~cn{b})"}
       {"cn{on}~cn{a} ==> rate{r} rate{t} cn{on}~cn{a}"
        &
        "rate{r} rate{t} cn{prot}~cn{b} ==> rate{r} rate{t} cn{prot}~cn{b}"}
     &
     \I{"rate{r} rate{t} rate{d} cn{off}~cn{a} ==> rate{r} rate{t} rate{d} cn{off}~cn{a}"}
    }
\end{gather*}

In this proof we are using the transition rules at many different worlds. This
is allowed because the rules are prefixed with "!!" and therefore available at
all worlds. Importantly, in the first premise of "P" we need to show that
"cn{on}~cn{a} ==> rate r rate t cn{on}~cn{a}". This is only possible if the rate
of a self-transition on "cn{on}~cn{a}" is "r . t". Of course, this is not
derivable from the rest of the theory (and may not actually be true), so it must
be added as a new rule; it is the contract that must be satisfied by the
repressilator in order for it to oscillate in the desired fashion.

All existing methods for modelling biology have algebraic foundations and 
none treats logic as the primary inferential device. 
In this section, we have sketched a mode of use of \hyll that lets one
represent the biological elements directly in the logic.
Note, however,
that unlike formalisms such as the brane or $\kappa$-calculi, we do not
propose \hyll as a new idealisation of biology. Instead, as far as systems
biology is concerned, our proposal should be seen as a uniform language to
encode biological systems; 
providing genuine means to reason about them is left for future work.

\section{Related work}
\label{sec:related}

Logically, the \hyll sequent calculus is a variant of labelled deduction, a very
broad topic not elaborated on here. The combination of linear logic with
labelled deduction isn't new to this work. In the
$\eta$-logic~\cite{deyoung08csf} the constraint domain is intervals of time, and
the rules of the logic generate constraint inequalities as a side-effect;
however its sole aim is the representation of proof-carrying authentication, and
it does not deal with genericity or focusing. The main feature of $\eta$ not in
\hyll is a separate constraint context that gives new constrained
propositions. \hyll is also related to the Hybrid Logical Framework
(HLF)~\cite{reed06hylo} which captures linear logic itself as a labelled form of
intuitionistic logic. Encoding constrained $\pi$ calculi directly in HLF would
be an interesting exercise: we would combine the encoding of linear logic with
the constraints of the process calculus. Because HLF is a very weak logic with a
proof theory based on natural deduction, it is not clear whether (and in what
forms) an adequacy result in \hyll can be transferred to HLF.


Temporal logics such as CSL and PCTL~\cite{hansson94fac} are popular
for logical reasoning on temporal properties of transition systems with probabilities.
In such logics, truth is defined in terms of correctness with respect to a
constrained forcing relation on the constraint algebra.
In CSL and PCTL
states are formal entities (names) labeled with atomic propositions.
Formulae are interpreted on algebraic structures that are discrete (in PCTL)
or continuous (in CSL) time Markov chains.
Transitions between states are viewed as couples of states labeled with a probability
(the probability of the transition), which is defined as a function
from $S \times S$ into $[0,1]$, where $S$ is the set of states.
While such logics have been very successful in practice with
efficient tools, the proof theory of these logics is very complex. Indeed, such
modal logics generally cannot be formulated in the sequent calculus, and
therefore lack cut-elimination and focusing. In contrast, \hyll has a very
traditional proof theoretic pedigree, but lacks such a close correspondence
between logical and algebraic equivalence. Probably the most well known and
relevant stochastic formalism not already discussed is that of stochastic
Petri-nets~\cite{marsan95book}, which have a number of sophisticated model
checking tools, including the PRISM framework~\cite{kwiatkowska04sttt}. Recent
advances in proof theory suggest that the benefits of model checking can be
obtained without sacrificing proofs and proof search~\cite{baelde07cade}.

\section{Conclusion and future work}
\label{sec:concl}

We have presented \hyll, a hybrid extension of intuitionistic linear logic with
a simple notion of situated truth, a traditional sequent calculus with
cut-elimination and focusing, and a modular and instantiable constraint system
(set of worlds) that can be directly manipulated using hybrid connectives.
We have proposed three instances of \hyll
(i.e three particular instances of the set of worlds):
one modelling temporal constraints and the others modelling probabilistic
or stochastic (continuous time Markov processes) constraints.
We have shown how to obtain representationally adequate encodings of constrained
transition systems, such as the synchronous stochastic $\pi$-calculus in a suitable
instance of \hyll.
We have also presented some preliminary experiments of direct encoding of biological 
systems, viewed as transition systems, in \hyll, using either temporal or probabilistic 
constraints.

Several instantiations of \hyll besides the ones in this paper seem
interesting. For example, we can already use disjunction ("plus") to explain
disjunctive states, but it is also possible to obtain a more extensional
branching by treating the worlds as points in an arbitrary partially-ordered set
instead of a monoid. Another possibility is to consider lists of worlds instead
of individual worlds -- this would allow defining periodic availability of a
resource, such as one being produced by an oscillating process. The most
interesting domain is that of discrete probabilities: here the underlying
semantics is given by discrete time Markov chains instead of CTMCs, which are
often better suited for symbolic evaluation~\cite{wu07qest}.

The logic we have provided so far is a logical framework well suited {\it to represent}
constrained transition systems. The design of a logical framework {\it for}
(i.e. to reason about) constrained transition systems is left for future work
-and might be envisioned by using a two-levels logical framework such as the Abella system.

An important open question is whether a general logic such as \hyll can serve as
a framework for specialized logics such as CSL and PCTL. A related question is
what benefit linearity truly provides for such logics -- linearity is obviously
crucial for encoding process calculi that are inherently stateful, but CSL
requires no such notion of single consumption of resources.

In the $\kappa$-calculus, reactions in a biological system are
modeled as reductions on graphs with certain state annotations.
It appears (though this has not been formalized)
that the $\kappa$-calculus can be embedded in \hyll even more naturally than
\spi, because a solution---a multiset of chemical products---is simply a tensor
of all the internal states of the binding sites together with the formed
bonds. One important innovation of $\kappa$ is the ability to extract
semantically meaningful ``stories'' from simulations. We believe that \hyll
provides a natural formal language for such stories.

We became interested in the problem of encoding stochastic reasoning in a
resource aware logic because we were looking for the logical essence of
biochemical reactions. What we envision for the domain of ``biological
computation'' is a resource-aware stochastic or probabilistic
$\lambda$-calculus that has \hyll propositions as (behavioral) types.
First step in this direction consists in exploiting and polishing 
the logic we have provided; This is the focus of our efforts at the CNRS, I3S.

\paragraph{Acknowledgements}

This work was partially supported by INRIA through the Information Society Technologies
programme of the European Commission, Future and Emerging Technologies under the
IST-2005-015905 MOBIUS project, and by the European TYPES project.
We thank Fran\c{c}ois Fages, Sylvain Soliman, Alessandra Carbone, Vincent Danos
and Jean Krivine for fruitful discussions on various preliminary versions of the
work presented here.
Thanks also go to Nicolas Champagnat and Luc Pronzato who helped us understand 
the algebraic nature of stochastic constraints.

\bgroup \small
\bibliographystyle{plain}
\bibliography{master,hyll-report,markov}

\begin{thebibliography}{10}

\bibitem{andreoli92jlc}
Jean-Marc Andreoli.
\newblock Logic programming with focusing proofs in linear logic.
\newblock {\em J. of Logic and Computation}, 2(3):297--347, 1992.

\bibitem{aziz00tcl}
A.~Aziz, K.~Sanwal, V.~Singhal, and R.~Brayton.
\newblock Model checking continuous time {M}arkov chains.
\newblock {\em ACM Transactions on Computational Logic}, 1(1):162--170, 2000.

\bibitem{baelde05stage}
David Baelde.
\newblock Logique lin\'eaire et alg\`ebre de processus.
\newblock Technical report, INRIA Futurs, LIX and ENS, 2005.

\bibitem{baelde07cade}
David Baelde, Andrew Gacek, Dale Miller, Gopalan Nadathur, and Alwen Tiu.
\newblock The {Bedwyr} system for model checking over syntactic expressions.
\newblock In F.~Pfenning, editor, {\em 21th Conf.\ on Automated Deduction
  (CADE)}, number 4603 in LNAI, pages 391--397, New York, 2007. Springer.

\bibitem{brauener06jal}
Torben Bra{\"u}ner and Valeria de~Paiva.
\newblock Intuitionistic hybrid logic.
\newblock {\em Journal of Applied Logic}, 4:231--255, 2006.

\bibitem{cardelli03bc}
Luca Cardelli.
\newblock Brane calculi.
\newblock In {\em Proceedings of BIO-CONCUR'03}, volume 180. Elsevier ENTCS,
  2003.

\bibitem{cervesato03tr}
Iliano Cervesato, Frank Pfenning, David Walker, and Kevin Watkins.
\newblock A concurrent logical framework {II}: Examples and applications.
\newblock Technical Report CMU-CS-02-102, Carnegie Mellon University, 2003.
\newblock Revised, May 2003.

\bibitem{chabrier05cmsb}
Nathalie Chabrier-Rivier, Fran\c{c}ois Fages, and Sylvain Soliman.
\newblock The biochemical abstract machine {BIOCHAM}.
\newblock In {\em International Workshop on Computational Methods in Systems
  Biology (CMSB-2)}, LNCS. Springer, 2004.

\bibitem{chaudhuri03tr}
Bor-Yuh~Evan Chang, Kaustuv Chaudhuri, and Frank Pfenning.
\newblock A judgmental analysis of linear logic.
\newblock Technical Report CMU-CS-03-131R, Carnegie Mellon University, December
  2003.

\bibitem{chaudhuri08jar}
Kaustuv Chaudhuri, Frank Pfenning, and Greg Price.
\newblock A logical characterization of forward and backward chaining in the
  inverse method.
\newblock {\em J. of Automated Reasoning}, 40(2-3):133--177, March 2008.

\bibitem{danos03bc}
Vincent Danos and Jean Krivine.
\newblock Formal molecular biology done in {CCS}.
\newblock In {\em Proceedings of BIO-CONCUR'03}, volume 180, pages 31--49.
  Elsevier ENTCS, 2003.

\bibitem{danos04tcs}
Vincent Danos and Cosimo Laneve.
\newblock Formal molecular biology.
\newblock {\em Theor. Comput. Sci.}, 325(1):69--110, 2004.

\bibitem{desharmais03jlap}
Jos\'{e}e Desharmais and Prakash Panangaden.
\newblock Continuous stochastic logic characterizes bisimulation of
  continuous-time {M}arkov processes.
\newblock {\em Journal of Logic and Algebraic Programming}, 56:99--115, 2003.

\bibitem{deyoung08csf}
Henry DeYoung, Deepak Garg, and Frank Pfenning.
\newblock An authorization logic with explicit time.
\newblock In {\em Computer Security Foundations Symposium (CSF-21)}, pages
  133--145. IEEE Computer Society, 2008.

\bibitem{elowitz00nature}
Michael~B. Elowitz and Stanislas Leibler.
\newblock A synthetic oscillatory network of transcriptional regulators.
\newblock {\em Nature}, 403(6767):335--338, 20 January 2000.

\bibitem{Emerson95}
E.~Allen Emerson.
\newblock Temporal and modal logic.
\newblock In {\em TCS}, pages 995--1072. Elsevier, 1995.

\bibitem{Ethier-Kurtz-book}
Stewart~N. Ethier and Thomas~G. Kurtz.
\newblock {\em Markov Processes; Characterization and Convergence}.
\newblock Wiley series in Probability and Mathematical Statistics.
  Wiley-interscience, 1986.

\bibitem{Foata-Fuchs-french-book}
Dominique Foata and Aim\'{e} Fuchs.
\newblock {\em Calcul des probabilit\'{e}s, cours exercices et probl\`{e}mes
  corrig\'{e}s}.
\newblock Dunod, 2e edition, 2003.

\bibitem{garg05concur}
Deepak Garg and Frank Pfenning.
\newblock Type-directed concurrency.
\newblock In Mart\'{\i}n Abadi and Luca de~Alfaro, editors, {\em 16th
  International Conference on Concurrency Theory (CONCUR)}, volume 3653 of {\em
  LNCS}, pages 6--20. Springer, 2005.

\bibitem{girard87tcs}
Jean-Yves Girard.
\newblock Linear logic.
\newblock {\em Theoretical Computer Science}, 50:1--102, 1987.

\bibitem{hansson94fac}
H.~Hansson and B.~Jonsson.
\newblock A logic for reasoning about time and probability.
\newblock {\em Formal Aspects of Computing}, (6), 1994.

\bibitem{hillston96book}
Jane Hillston.
\newblock {\em A compositional approach to performance modelling}.
\newblock Cambridge University Press, 1996.

\bibitem{kamp68phd}
Johan Anthony~Willem Kamp.
\newblock {\em Tense Logic and the Theory of Linear Order}.
\newblock PhD thesis, University of California, Los Angeles, 1968.

\bibitem{kwiatkowska04sttt}
M.~Kwiatkowska, G.~Norman, and D.~Parker.
\newblock Probabilistic symbolic model checking using {PRISM}: a hybrid
  approach.
\newblock {\em International Journal of Software Tools for Technology
  Transfer}, 6(2), 2004.

\bibitem{marsan95book}
M.~Ajmone Marsan, G.~Balbo, G.~Conte, S.~Donatelli, and G.~Francenschinis.
\newblock {\em Modelling with Generalised Stochastic {P}etri Nets}.
\newblock Wiley Series in Parallel Computing. Wiley and Sons, 1995.

\bibitem{miller92welp}
Dale Miller.
\newblock The $\pi$-calculus as a theory in linear logic: Preliminary results.
\newblock In E.~Lamma and P.~Mello, editors, {\em 3rd Workshop on Extensions to
  Logic Programming}, number 660 in LNCS, pages 242--265, Bologna, Italy, 1993.
  Springer.

\bibitem{milner99book}
Robin Milner.
\newblock {\em Communicating and Mobile Systems : The $\pi$-Calculus}.
\newblock Cambridge University Press, New York, NY, USA, 1999.

\bibitem{nigam08ijcar}
Vivek Nigam and Dale Miller.
\newblock Focusing in linear meta-logic.
\newblock In {\em Proceedings of IJCAR: International Joint Conference on
  Automated Reasoning}, volume 5195 of {\em LNAI}, pages 507--522. Springer,
  2008.

\bibitem{pfenning88pldi}
Frank Pfenning and Conal Elliott.
\newblock Higher-order abstract syntax.
\newblock In {\em Proceedings of the {ACM}-{SIGPLAN} Conference on Programming
  Language Design and Implementation}, pages 199--208. ACM Press, June 1988.

\bibitem{phillips04cmmb}
Andrew Phillips and Luca Cardelli.
\newblock A correct abstract machine for the stochastic pi-calculus.
\newblock {\em Concurrent Models in Molecular Biology}, August 2004.

\bibitem{phillips04bc}
Andrew Phillips and Luca Cardelli.
\newblock A correct abstract machine for the stochastic pi-calculus.
\newblock In {\em Proceedings of BioConcur'04}, ENTCS, 2004.

\bibitem{phillips06tcsb}
Andrew Phillips, Luca Cardelli, and Giuseppe Castagna.
\newblock A graphical representation for biological processes in the stochastic
  pi-calculus.
\newblock {\em Transactions on Computational Systems Biology VII}, pages
  123--152, 2006.

\bibitem{prior57book}
Arthur~N. Prior.
\newblock {\em Time and Modality}.
\newblock Oxford: Clarendon Press, 1957.

\bibitem{reed06hylo}
Jason Reed.
\newblock Hybridizing a logical framework.
\newblock In {\em International Workshop on Hybrid Logic ({HyLo})}, Seattle,
  USA, August 2006.

\bibitem{regev04tcs}
A.~Regev, E.~M. Panina, W.~Silverman, L.~Cardelli, and E.~Shapiro.
\newblock Bioambients: an abstraction for biological compartments.
\newblock {\em Theoretical Computer Science}, 325(1):141--167, 2004.

\bibitem{Regev01psb}
A.~Regev, W.~Silverman, and E.~Shapiro.
\newblock Representation and simulation of biochemical processes using the
  {$\pi$}-calculus and process algebra.
\newblock In L.~Hunter {R. B. Altman, A. K. Dunker} and T.~E. Klein, editors,
  {\em Pacific Symposium on Biocomputing}, volume~6, pages 459--470, Singapore,
  2001. World Scientific Press.

\bibitem{Rogers-Williams-vol1-book}
L.~C.~G. Rogers and D.~Williams.
\newblock {\em Diffusions, Markov Processes and Martingales}, volume 1:
  Foundations.
\newblock Cambridge Mathematical Library, 2nd edition, 2000.

\bibitem{saranli07icra}
Ulu{\c{c}} Saranli and Frank Pfenning.
\newblock Using constrained intuitionistic linear logic for hybrid robotic
  planning problems.
\newblock In {\em IEEE International Conference on Robotics and Automation
  (ICRA)}, pages 3705--3710. IEEE, 2007.

\bibitem{simpson94phd}
Alex Simpson.
\newblock {\em The Proof Theory and Semantics of Intuitionistic Modal Logic}.
\newblock PhD thesis, University of Edinburgh, 1994.

\bibitem{watkins03tr}
Kevin Watkins, Iliano Cervesato, Frank Pfenning, and David Walker.
\newblock A concurrent logical framework {I}: Judgments and properties.
\newblock Technical Report CMU-CS-02-101, Carnegie Mellon University, 2003.
\newblock Revised, May 2003.

\bibitem{wu07qest}
Peng Wu, Catuscia Palamidessi, and Huimin Lin.
\newblock Symbolic bisimulations for probabilistic systems.
\newblock In {\em QEST'07}, pages 179--188. IEEE Computer Society, 2007.

\end{thebibliography}
\egroup

\clearpage
\appendix

\section{Proofs}
\label{sec:proofs}

\subsection{Identity principle}
\label{sec:proofs.identity}

\begin{thm}[Identity principle] The following rule is derivable.
  \begin{gather*}
    \I[init*]{"\G ; A @ w ==> A @ w"}{}
  \end{gather*}
\end{thm}

\begin{proof}
  By induction on the structure of "A". We have the following cases.
  \begin{ecom} [{\itshape {case}}] 
  \item "A" is an atom "p~\vec t". Then, "\G ; p~\vec t @ w ==> p~\vec t @ w" by init.
  \item "A" is "B with C".
    \begin{gather*}
      \I["with R"]{"\G ; B with C @ w ==> B with C @ w"}
        {\I["with L_1"]{"\G ; B with C @ w ==> B @ w"}
           {\I[i.h.]{"\G ; B @ w ==> B @ w"}{}}
         &
         \I["with L_2"]{"\G ; B with C @ w ==> C @ w"}
           {\I[i.h.]{"\G ; C @ w ==> C @ w"}{}}
        }
    \end{gather*}
  \item "A" is "top".
    \begin{gather*}
      \I["top R"]{"\G ; top @ w ==> top @ w"}{}
    \end{gather*}
  \item "A" is "B plus C".
    \begin{gather*}
      \I["plus L"]{"\G ; B plus C @ w ==> B plus C @ w"}
        {\I["plus R_1"]{"\G ; B @ w ==> B plus C @ w"}
           {\I[i.h.]{"\G ; B @ w ==> B @ w"}{}}
         &
         \I["plus R_2"]{"\G ; C @ w ==> B plus C @ w"}
           {\I[i.h.]{"\G ; C @ w ==> C @ w"}{}}
        }
    \end{gather*}
  \item "A" is "zero".
    \begin{gather*}
      \I["zero L"]{"\G ; zero @ w ==> zero @ w"}{}
    \end{gather*}
  \item "A" is "B -o C".
    \begin{gather*}
      \I["{-o} R"]{"\G ; B -o C @ w ==> B -o C @ w"}
        {\I["{-o}L"]{"\G ; B -o C @ w, B @ w ==> C @ w"}
           {\I[i.h.]{"\G ; B @ w ==> B @ w"}{}
            &
            \I[i.h.]{"\G ; C @ w ==> C @ w"}{}
           }
        }
    \end{gather*}
  \item "A" is "B tens C".
    \begin{gather*}
      \I["tens L"]{"\G ; B tens C @ w ==> B tens C @ w"}
        {\I["tens R"]{"\G ; B @ w, C @ w ==> B tens C @ w"}
           {\I[i.h.]{"\G ; B @ w ==> B @ w"}{}
            &
            \I[i.h.]{"\G ; C @ w ==> C @ w"}{}
           }
        }
    \end{gather*}
  \item "A" is "one".
    \begin{gather*}
      \I["one L"]{"\G ; one @ w ==> one @ w"}
        {\I["one R"]{"\G ; . ==> one @ w"}{}}
    \end{gather*}
  \item "A" is "all x. B".
    \begin{gather*}
      \I["\forall R^\alpha"]{"\G ; \fall \alpha B @ w ==> \fall \alpha B @ w"}
        {\I["\forall L"]{"\G ; \fall \alpha B @ w ==> B @ w"}
           {\I[i.h.]{"\G ; B @ w ==> B @ w"}}
        }
    \end{gather*}
  \item "A" is "ex x. B".
    \begin{gather*}
      \I["\exists L^\alpha"]{"\G ; \fex \alpha B @ w ==> \fex \alpha B @ w"}
        {\I["\exists R"]{"\G ; B @ w ==> \fex \alpha B @ w"}
           {\I[i.h.]{"\G ; B @ w ==> B @ w"}}
        }
    \end{gather*}
  \item "A" is "! B".
    \begin{gather*}
      \I["!L"]{"\G ; {! B} @ w ==> {! B} @ w"}
        {\I["!R"]{"\G, B @ w ; . ==> {! B} @ w"}
           {\I[copy]{"\G, B @ w ; . ==> B @ w"}
              {\I[i.h.]{"\G, B @ w ; B @ w ==> B @ w"}{}}
           }
        }
    \end{gather*}
  \item "A" is "now u. B".
    \begin{gather*}
      \I["{dn} R"]{"\G ; now u. B @ w ==> now u. B @ w"}
        {\I["{dn} L"]{"\G ; now u. B @ w ==> [w/u] B @ w"}
           {\I[i.h.]{"\G ; [w/u] B @ w ==> [w/u] B @ w"}{}}
        }
    \end{gather*}
  \item "A" is "(B at v)".
    \begin{gather*}
      \I["at R"]{"\G ; (B at v) @ w ==> (B at v) @ w"}
        {\I["at L"]{"\G ; (B at v) @ w ==> B @ v"}
           {\I[i.h.]{"\G ; B @ v ==> B @ v"}{}}}
    \end{gather*}
  \end{ecom}
\end{proof}

\subsection{Cut admissibility}
\label{sec:proofs.cut}

\begin{thm}[Cut admissibility] The following two rules are admissible.
  \begin{gather*}
    \I[cut]{"\G ; \D, \D' ==> C @ {w'}"}
      {"\G ; \D ==> A @ w" & "\G ; \D', A @ w ==> C @ {w'}"}
    \\
    \I[cut!]{"\G ; \D ==> C @ {w'}"}
      {"\G ; . ==> A @ w" & "\G, A @ w ; \D ==> C @ {w'}"}
  \end{gather*}
\end{thm}

\begin{proof}
  Name the two premise derivations 
  "\DD" and "\EE" respectively. The proof proceeds by induction on
  the structure of the derivations "\DD" and "\EE", and more precisely on
  a lexicographic order that allows the induction hypothesis to be used whenever:
  \begin{ecom}
  \item The cut formula becomes strictly smaller (in the subformula relation), or
  \item The cut formula remains the same, but an instance of cut is used to justify an instance of cut!.
  \item The cut formula remains the same, but the derivation "\DD" is strictly smaller, or
  \item The cut formula remains the same, but the derivation "\EE" is strictly smaller, or
  \end{ecom}
  In each case, we consider derivations to be identical that differ in such a way that one can be
  derived from the other simply by weakening and contracting the unrestricted contexts of their
  respective sequents. The lexicographic order is well-founded because the given derivations "\DD"
  and "\EE" are finite, and cut! is used at most once per subformula of "A" (see ``copy cuts''
  below). All the cuts break down into the following four major categories.

  \paragraph{Atomic cuts} where the formula "A" is an atom "p~(\vec t)". We have the following two
  cases;
  \begin{ecom}  [{\itshape {Case}.}]  
  \item "\DD" is:
    \begin{gather*}
      \I[init]{"\G ; p~(\vec t) @ w ==> p~(\vec t) @ w"}{}
    \end{gather*}
    Then the result of the cut has the same conclusion as that of "\EE".

  \item "\EE" is
    \begin{gather*}
      \I[init]{"\G ; p~(\vec t) @ w ==> p~(\vec t) @ w"}{}
    \end{gather*}
    Then the result of the cut has the same conclusion as that of "\DD".
  \end{ecom}

  \paragraph{Principal cuts} where a non-atomic cut formula "A" is introduced by a final right rule
  in "\DD" and a final left-rule in "\EE". We have the following cases.
  \begin{ecom} [{\itshape {Case}.}]  
  \item "A" is "A_1 with A_2", and:
    \begin{gather*}
      \DD =
      \Ic["with R"]{"\G ; \D ==> A_1 with A_2 @ w"}
         {"\DD_1 :: \G ; \D ==> A_1 @ w" & "\DD_2 :: \G ; \D ==> A_2 @ w"}
      \SP
      \EE =
      \Ic["with L_i"]{"\G ; \D', A_1 with A_2 @ w ==> C @ {w'}"}
         {"\EE' :: \G ; \D', A_i @ w ==> C @ {w'}"}
    \end{gather*}
    Then:
    \begin{quote}
      "\G ; \D, \D' ==> C @ {w'}" \by cut on "\DD_i" and "\EE'".
    \end{quote}
  \item "A" is "A_1 plus A_2", and:
    \begin{gather*}
      \lift{\DD =\ }
      \I["plus R_i"]{"\G ; \D ==> A_1 plus A_2 @ w"}
         {"\DD' :: \G ; \D ==> A_i @ w"}
      \SP
      \lift{\EE =\ }
      \I["plus L"]{"\G ; \D', A_1 plus A_2 @ w ==> C @ {w'}"}
         {\begin{array}[b]{c}
             "\EE_1 :: \G ; \D', A_1 @ w ==> C @ {w'}" \\
             "\EE_2 :: \G ; \D', A_2 @ w ==> C @ {w'}"
          \end{array}}
    \end{gather*}
    Then:
    \begin{quote}
      "\G ; \D, \D' ==> C @ {w'}" \by cut on "\DD'" and "\EE_i".
    \end{quote}
  \item "A" is "A_1 -o A_2", and:
    \begin{gather*}
      \DD =
      \Ic["{-o} R"]{"\G ; \D ==> A_1 -o A_2 @ w"}
         {"\DD' :: \G ; \D, A_1 @ w ==> A_2 @ w"}
      \SP
      \EE =
      \Ic["{-o}L"]{"\G ; \D'_1, \D'_2, A_1 -o A_2 ==> C @ {w'}"}
         {"\EE_1 :: \G ; \D'_1 ==> A_1 @ w"
          &
          "\EE_2 :: \G ; \D'_2, A_2 @ w ==> C @ {w'}"}
    \end{gather*}
    Then:
    \begin{quote}
      "\G ; \D, A_1 @ w, \D'_2 ==> C @ {w'}" \by cut on "\DD'" and "\EE_2". \\
      "\G ; \D, \D'_1, \D'_2 ==> C @ {w'}" \by cut on "\EE_1" and above.
    \end{quote}

  \item "A" is "A_1 tens A_2", and:
    \begin{gather*}
      \DD =
      \Ic["tens R"]{"\G ; \D_1, \D_2 ==> A_1 tens A_2 @ w"}
         {"\DD_1 :: \G ; \D_1 ==> A_1 @ w" & "\DD_2 :: \G ; \D_2 ==> A_2 @ w"}
      \SP
      \EE =
      \Ic["tens L"]{"\G ; \D', A_1 tens A_2 @ w ==> C @ {w'}"}
         {"\EE' :: \G ; \D', A_1 @ w, A_2 @ w ==> C @ {w'}"}
    \end{gather*}
    Then:
    \begin{quote}
      "\G ; \D', \D_2, A_1 @ w ==> C @ {w'}" \by cut on "\DD_2" and "\EE'".\\
      "\G ; \D', \D_1, \D_2 ==> C @ {w'}" \by cut on "\DD_1" and above.
    \end{quote}
  \item "A" is "one", and:
    \begin{gather*}
      \DD =
      \Ic["one R"]{"\G ; . ==> one @ w"}
      \SP
      \EE =
      \Ic["one L"]{"\G ; \D', one @ w ==> C @ {w'}"}
         {"\EE' :: \G ; \D' ==> C @ {w'}"}
    \end{gather*}
    The result of the cut is the conclusion of "\EE'".

  \item "A" is "all x. B", and:
    \begin{gather*}
      \DD =
      \Ic["\forall R^\alpha"]{"\G ; \D ==> \fall \alpha B @ w"}
         {\DD'(\alpha) :: "\G ; \D ==> B @ w"}
      \SP
      \EE =
      \Ic["\forall L"]{"\G ; \D', \fall \alpha B @ w ==> C @ {w'}"}
         {\EE' :: "\G ; \D', [\tau / \alpha] B @ w ==> C @ {w'}"}
    \end{gather*}
    Let "a" be any parameter. Then:
    \begin{quote}
      "\G ; \D, \D' ==> C @ {w'}" \by cut on "\DD'(\tau)" and "\EE'".
    \end{quote}

  \item "A" is "ex x. B", and:
    \begin{gather*}
      \DD =
      \Ic["\exists R"]{"\G ; \D ==> \fex \alpha B @ w"}
         {\DD' :: "\G ; \D ==> [\tau / \alpha] B @ w"}
      \SP
      \EE =
      \Ic["\exists L^\alpha"]{"\G ; \D', \fex \alpha B @ w ==> C @ {w'}"}
         {\EE'(\alpha) :: "\G ; \D', B @ w ==> C @ {w'}"}
    \end{gather*}
    Let "a" be any parameter. Then:
    \begin{quote}
      "\G ; \D, \D' ==> C @ {w'}" \by cut on "\DD'" and "\EE'(\alpha)".
    \end{quote}

  \item "A" is "! B", and:
    \begin{gather*}
      \DD =
      \Ic["!R"]{"\G ; . ==> {! B} @ w"}
         {\DD' :: "\G ; . ==> B @ w"}
      \SP
      \EE =
      \Ic["!L"]{"\G ; \D', {! B} @ w ==> C @ {w'}"}
         {\EE' :: "\G, B @ w ; \D' ==> C @ {w'}"}
    \end{gather*}
    Then:
    \begin{quote}
      "\G ; \D' ==> C @ {w'}" \by cut! on "\DD'" and "\EE'".
    \end{quote}

  \item "A" is "now u. B", and:
    \begin{gather*}
      \DD =
      \Ic["{dn} R"]{"\G ; \D ==> now u. B @ w"}
         {"\DD' :: \G ; \D ==> [w/u] B @ w"}
      \SP
      \EE =
      \Ic["{dn} L"]{"\G ; \D', now u. B @ w ==> C @ {w'}"}
         {"\EE' :: \G ; \D', [w/u] B @ w ==> C @ {w'}"}
    \end{gather*}
    Then:
    \begin{quote}
      "\G ; \D, \D' ==> C @ {w'}" \by cut on "\DD'" and "\EE'".
    \end{quote}

  \item "A" is "(B at v)", and:
    \begin{gather*}
      \DD =
      \Ic["at R"]{"\G ; \D ==> (B at v) @ w"}
         {"\DD' :: \G ; \D ==> B @ v"}
      \SP
      \EE =
      \Ic["at L"]{"\G ; \D', (B at v) @ w ==> C @ {w'}"}
         {"\EE' :: \G ; \D', B @ v ==> C @ {w'}"}
    \end{gather*}
    Then:
    \begin{quote}
      "\G ; \D, \D' ==> C @ {w'}" \by cut on "\DD'" and "\EE'".
    \end{quote}
  \end{ecom}

  \paragraph{Copy cuts} where the cut formula in "\EE" was transferred using copy, i.e.:
  \begin{gather*}
    \DD :: "\G ; . ==> A @ w"
    \SP
    \EE =
    \Ic[copy]{"\G, A @ w ; \D' ==> C @ {w'}"}
       {"\EE' :: \G, A @ w ; \D', A @ w ==> C @ {w'}"}
  \end{gather*}
  Here,
  \begin{quote}
    "\G, A @ w ; . ==> A @ w" \by weakening on "\DD".\\
    "\G, A @ w ; \D' ==> C @ {w'}" \by cut on "\DD" and "\EE'".\\
    "\G ; \D' ==> C @ {w'}" \by cut! on "\DD" and above.
  \end{quote}
  The first cut is applied on a variant of "\DD" that differs from "\DD" only in terms of a weaker
  unrestricted context. In the last step, a cut was used to justify a cut!, which is allowed by the
  lexicographic order.

  \paragraph{Left-commutative cuts} where the cut formula "A" is a side formula in the derivation
  "\DD". The following is a representative case.
  \begin{gather*}
    \DD =
    \Ic["tens L"]{"\G ; \D, D tens E @ {w''} ==> A @ w"}
       {\DD' :: "\G ; \D, D @ {w''}, E @ {w''} ==> A @ w"}
    \SP
    \EE :: "\G ; \D', A @ w ==> C @ {w'}".
  \end{gather*}
  Here,
  \begin{quote}
    "\G ; \D, D @ {w''}, E @ {w''}, \D' ==> C @ {w'}" \by cut on "\DD'" and "\EE". \\
    "\G ; \D, \D', D tens E @ {w''} ==> C @ {w'}" \by "tens L".
  \end{quote}

  \paragraph{Right-commutative cuts} where the cut formula "A" is a side formula in the derivation
  "\EE". The following is a representative case.
  \begin{gather*}
    \DD :: "\G ; \D ==> A @ w"
    \SP
    \EE =
    \Ic["with R"]{"\G ; \D', A @ w ==> D with E @ {w'}"}
       {\EE_1 :: "\G ; \D', A @ w ==> D @ {w'}" & \EE_2 :: "\G ; \D', A @ w ==> E @ {w'}"}
  \end{gather*}
  Here,
  \begin{quote}
    "\G ; \D, \D' ==> D @ {w'}" \by cut on "\DD" and "\EE_1". \\
    "\G ; \D, \D' ==> E @ {w'}" \by cut on "\DD" and "\EE_2". \\
    "\G ; \D, \D' ==> D with E @ {w'}" \by "with R".
  \end{quote}

  \noindent
  This completes the inventory of all possible cuts.
\end{proof}

\subsection{Invertibility}
\label{sec:proofs.invert}

\begin{thm}[Invertibility] The following rules are invertible:
  \begin{ecom}
  \item On the right: "with R", "top R", "{-o} R", "\forall R", "{dn} R" and "@ R";
  \item On the left: "tens L", "one L", "plus L", "zero L", "\exists L", "!L", "{dn} L" and "at L".
  \end{ecom}
\end{thm}

\begin{proof}
  Each inversion is shown to be admissible using a suitable cut.
  \begin{ecom} [{\itshape {Case of}}]  
  \item "with R":
    \begin{gather*}
      \I[cut]{"\G ; \D ==> A_i @ w"}
        {"\G ; \D ==> A_1 with A_2 @ w"
         &
         \I["with L_i"]{"\G ; A_1 with A_2 @ w ==> A_i @ w"}
           {\I[init*]{"\G ; A_i @ w ==> A_i @ w"}{}}
        }
    \end{gather*}
  \item "top R": trivial.
  \item "{-o} R":
    \begin{gather*}
      \I[cut]{"\G ; \D, A @ w ==> B @ w"}
        {"\G ; \D ==> A -o B @ w"
         &
         \I["{-o} L"]{"\G ; A -o B @ w, A @ w ==> B @ w"}
           {\I[init*]{"\G ; A @ w ==> A @ w"}{}
            &
            \I[init*]{"\G ; B @ w ==> B @ w"}{}}
        }
    \end{gather*}
  \item "\forall R":
    \begin{gather*}
      \I[cut]{"\G ; \D ==> A @ w"}
        {"\G ; \D ==> \fall \alpha A @ w"
         &
         \I["\forall L"]{"\G ; \fall \alpha A @ w ==> A @ w"}
           {\I[init*]{"\G ; A @ w ==> A @ w"}}
        }
    \end{gather*}
  \item "{dn} R":
    \begin{gather*}
      \I[cut]{"\G ; \D ==> [w/u] A @ w"}
        {"\G ; \D ==> now u. A @ w"
         &
         \I["{dn} L"]{"\G ; now u. A @ w ==> [w/u] A @ w"}
           {\I[init*]{"\G ; [w/u] A @ w ==> [w/u] A @ w"}{}}
        }
    \end{gather*}
  \item "at R":
    \begin{gather*}
      \I[cut]{"\G ; \D ==> A @ v"}
        {"\G ; \D ==> (A at v) @ w"
         &
         \I["at L"]{"\G ; (A at v) @ w ==> A @ v"}
           {\I[init*]{"\G ; A @ v ==> A @ v"}{}}
        }
    \end{gather*}

  \item "tens L":
    \begin{gather*}
      \I[cut]{"\G ; \D, A @ w, B @ w ==> C @ {w'}"}
        {\I["tens R"]{"\G ; A @ w, B @ w ==> A tens B @ w"}
           {\I[init*]{"\G ; A @ w ==> A @ w"}{}
            &
            \I[init*]{"\G ; B @ w ==> B @ w"}{}}
         &
         "\G ; \D, A tens B @ w ==> C @ {w'}"}
    \end{gather*}

  \item "one L":
    \begin{gather*}
      \I[cut]{"\G ; \D ==> C @ {w'}"}
        {\I["one R"]{"\G ; . ==> one @ w"}{}
         &
         "\G ; \D, one @ w ==> C @ {w'}"}
    \end{gather*}

  \item "plus L":
    \begin{gather*}
      \I[cut]{"\G ; \D, A_i @ w ==> C @ {w'}"}
        {\I["plus R_i"]{"\G ; A_i @ w ==> A_1 plus A_2 @ w"}
           {\I[init*]{"\G ; A_i @ w ==> A_i @ w"}{}}
         &
         "\G ; \D, A_1 plus A_2 @ w ==> C @ {w'}"}
    \end{gather*}

  \item "zero L": trivial.

  \item "\exists L":
    \begin{gather*}
      \I[cut]{"\G ; \D, A @ w ==> C @ {w'}"}
        {\I["\exists R"]{"\G ; A @ w ==> \fex \alpha A @ w"}
           {\I[init*]{"\G ; A @ w ==> A @ w"}}
         &
         "\G ; \D, ex x. A @ w ==> C @ {w'}"}
    \end{gather*}

  \item "!L":
    \begin{gather*}
      \I[cut]{"\G, A @ w ; \D ==> C @ {w'}"}
        {\I["!R"]{"\G, A @ w ; . ==> {! A} @ w"}
           {\I[copy]{"\G, A @ w ; . ==> A @ w"}
              {\I[init*]{"\G, A @ w ; A @ w ==> A @ w"}{}}
           }
         &
         \I[weaken]{"\G, A @ w ; \D, {! A} @ w ==> C @ {w'}"}
           {"\G ; \D, {! A} @ w ==> C @ {w'}"}}
    \end{gather*}

  \item "{dn}L":
    \begin{gather*}
      \I[cut]{"\G ; \D, [w/u] A @ w ==> C @ {w'}"}
        {\I["{dn}R"]{"\G ; [w/u] A @ w ==> now u. A @ w"}
           {\I[init*]{"\G ; [w/u] A @ w ==> [w/u] A @ w"}{}}
         &
         "\G ; \D, now u. A @ w ==> C @ {w'}"}
    \end{gather*}

  \item "at L":
    \begin{gather*}
      \I[cut]{"\G ; \D, A @ v ==> C @ {w'}"}
        {\I["at R"]{"\G ; A @ v ==> (A at v) @ w"}
           {\I[init*]{"\G ; A @ v ==> A @ v"}{}}
         &
         "\G ; \D, (A at v) @ w ==> C @ {w'}"}
    \end{gather*}
  \end{ecom}
\end{proof}

\subsection{Correctness and consistency}
\label{sec:proofs.correct}

\begin{thm}[Correctness of the sequent calculus] \mbox{}
  \begin{ecom}
  \item If "\G ; \D ==> C @ w", then "\G ; \D |-nd C @ w". (soundness)
  \item If "\G ; \D |-nd C @ w", then "\G ; \D ==> C @ w". (completeness)
  \end{ecom}
\end{thm}

\begin{proof}
  The right rules of the sequent calculus and the introduction rules of natural deduction
  coincide. Therefore, for (1), we need only to show that the judgemental and left rules of the
  sequent calculus are admissible in natural deduction, and for (2), only to show that the
  judgemental and elimination rules of natural deduction are admissible in the sequent calculus. The
  following are the main cases.
  \begin{ecom} ["==>"/"|-nd" {case}.]  
  \item (init) 
    \begin{gather*}
      \I[hyp]{"\G ; p~(\vec t) @ w |-nd p~(\vec t) @ w"}{}
    \end{gather*}
  \item (copy)
    \begin{gather*}
      \I[subst]{"\G, A @ w ; \D |-nd C @ {w'}"}
        {\I[hyp!]{"\G, A @ w ; . |-nd A @ w"}{}
         &
         "\G, A @ w ; \D, A @ w |-nd C @ {w'}"
        }
    \end{gather*}
  \item ("with L_i")
    \begin{gather*}
      \I[subst]{"\G ; \D, A_1 with A_2 @ w |-nd C @ {w'}"}
        {\I["with E_i"]{"\G ; A_1 with A_2 @ w |-nd A_i @ w"}
           {\I[hyp]{"\G ; A_1 with A_2 @ w |-nd A_1 with A_2 @ w"}{}}
         &
         "\G ; \D, A_i @ w |-nd C @ {w'}"
        }
    \end{gather*}
  \item ("plus L")
    \begin{gather*}
      \I["plus E"]{"\G ; \D, A_1 plus A_2 @ w |-nd C @ {w'}"}
        {\I[hyp]{"\G ; A_1 plus A_2 @ w |-nd A_1 plus A_2 @ w"}{}
         &
         "\G ; \D, A_1 @ w |-nd C @ {w'}"
         &
         "\G ; \D, A_2 @ w |-nd C @ {w'}"
        }
    \end{gather*}
  \item ("zero L")
    \begin{gather*}
      \I["zero E"]{"\G ; \D, zero @ w |-nd C @ {w'}"}
        {\I[hyp]{"\G ; zero @ w |-nd zero @ w"}}
    \end{gather*}
  \item ("tens L")
    \begin{gather*}
      \I["tens E"]{"\G ; \D, A tens B @ w |-nd C @ {w'}"}
        {\I[hyp]{"\G ; A tens B @ w |-nd A tens B @ w"}{}
         &
         "\G ; \D, A @ w, B @ w |-nd C @ {w'}"
        }
    \end{gather*}
  \item ("one L")
    \begin{gather*}
      \I["one E"]{"\G ; \D, one @ w |-nd C @ {w'}"}
        {\I[hyp]{"\G ; one @ w |-nd one @ w"}{}
         &
         "\G ; \D |-nd C @ {w'}"}
    \end{gather*}
  \item ("{-o} L")
    \begin{gather*}
      \I[subst]{"\G ; \D, \D', A -o B @ w |-nd C @ {w'}"}
        {\I["-o E"]{"\G ; A -o B @ w |-nd B @ w"}
           {\I[hyp]{"\G ; A -o B @ w |-nd A -o B @ w"}{}
            &
            "\G ; \D |-nd A @ w"
           }
         &
         "\G ; \D', B @ w |-nd C @ {w'}"
        }
    \end{gather*}
  \item ("\forall L")
    \begin{gather*}
      \I[subst]{"\G ; \D, \fall \alpha A @ w |-nd C @ {w'}"}
        {\I["\forall E"]{"\G ; \fall \alpha A @ w |-nd [\tau / \alpha] A @ w"}
           {\I[hyp]{"\G ; \fall \alpha A @ w |-nd \fall \alpha A @ w"}{}}
         &
         "\G ; \D, [\tau / \alpha] A @ w |-nd C @ {w'}"
        }
    \end{gather*}
  \item ("\exists L")
    \begin{gather*}
      \I["\exists E^\alpha"]{"\G ; \D, \fex \alpha A @ w |-nd C @ {w'}"}
        {\I[hyp]{"\G ; \fex \alpha A @ w |-nd \fex \alpha A @ w"}{}
         &
         "\G ; \D, A @ w |-nd C @ {w'}"
        }
    \end{gather*}
  \item ("!L")
    \begin{gather*}
      \I["!E"]{"\G ; \D, {! A} @ w |-nd C @ {w'}"}
        {\I[hyp]{"\G ; {! A} @ w |-nd {! A} @ w"}{}
         &
         "\G, {! A} @ w ; \D |-nd C @ {w'}"}
    \end{gather*}
  \item ("{dn} L")
    \begin{gather*}
      \I[subst]{"\G ; \D, now u. A @ w |-nd C @ {w'}"}
        {\I["{dn} E"]{"\G ; now u. A @ w |-nd [w/u] A @ w"}
           {\I[hyp]{"\G ; now u. A @ w |-nd now u. A @ w"}{}}
         &
         "\G ; \D, [w/u] A @ w |-nd C @ {w'}"}
    \end{gather*}
  \item ("at~ L")
    \begin{gather*}
      \I[subst]{"\G ; \D, (A at v) @ w |-nd C @ {w'}"}
        {\I["at~ E"]{"\G ; (A at v) @ w |-nd A @ v"}
           {\I[hyp]{"\G ; (A at v) @ w |-nd (A at v) @ w"}{}}
         &
         "\G ; \D, A @ v |-nd C @ {w'}"}
    \end{gather*}
  \end{ecom}

  \begin{ecom} ["|-nd"/"==>" {case}.]  
  \item (hyp)
    \begin{gather*}
      \I[init*]{"\G ; A @ w ==> A @ w"}{}
    \end{gather*}
  \item (hyp!)
    \begin{gather*}
      \I[copy]{"\G, A @ w ; . ==> A @ w"}
        {\I[init*]{"\G, A @ w ; A @ w ==> A @ w"}{}}
    \end{gather*}
  \item ("with E_i")
    \begin{gather*}
      \I[cut]{"\G ; \D ==> A_i @ w"}
        {"\G ; \D ==> A_1 with A_2 @ w"
         &
         \I["with L_i"]{"\G ; A_1 with A_2 @ w ==> A_i @ w"}
           {\I[init*]{"\G ; A_i @ w ==> A_i @ w"}{}}
        }
    \end{gather*}
  \item ("plus E")
    \begin{gather*}
      \I[cut]{"\G ; \D, \D' ==> C @ {w'}"}
        {"\G ; \D ==> A plus B @ w"
         &
         \I["plus L"]{"\G ; \D', A plus B @ w ==> C @ {w'}"}
           {"\G ; \D', A @ w ==> C @ {w'}"
            &
            "\G ; \D', B @ w ==> C @ {w'}"
           }
        }
    \end{gather*}
  \item ("zero E")
    \begin{gather*}
      \I[cut]{"\G ; \D, \D' ==> C @ {w'}"}
        {"\G ; \D ==> zero @ w"
         &
         \I["zero L"]{"\G ; \D', zero @ w ==> C @ {w'}"}{}
        }
    \end{gather*}
  \item ("tens E")
    \begin{gather*}
      \I[cut]{"\G ; \D, \D' ==> C @ {w'}"}
        {"\G ; \D ==> A tens B @ w"
         &
         \I["tens L"]{"\G ; \D', A tens B @ w ==> C @ {w'}"}
           {"\G ; \D', A @ w, B @ w ==> C @ {w'}"}
        }
    \end{gather*}
  \item ("one E")
    \begin{gather*}
      \I[cut]{"\G ; \D, \D' ==> C @ {w'}"}
        {"\G ; \D ==> one @ w"
         &
         \I["one L"]{"\G ; \D', one @ w ==> C @ {w'}"}
           {"\G ; \D' ==> C @ {w'}"}
        }
    \end{gather*}
  \item ("\forall E")
    \begin{gather*}
      \I[cut]{"\G ; \D ==> [\tau / \alpha] A @ w"}
        {"\G ; \D ==> \fall \alpha A @ w"
         &
         \I["\forall L"]{"\G ; \fall \alpha A @ w ==> [\tau / \alpha] A @ w"}
           {\I[init*]{"\G ; [\tau / \alpha] A @ w ==> [\tau / \alpha] A @ w"}{}}
        }
    \end{gather*}
  \item ("\exists E")
    \begin{gather*}
      \I[cut]{"\G ; \D, \D' ==> C @ {w'}"}
        {"\G ; \D ==> \fex \alpha A @ w"
         &
         \I["\exists L^\alpha"]{"\G ; \D', \fex \alpha A @ w ==> C @ {w'}"}
           {"\G ; \D', A @ w ==> C @ {w'}"}
        }
    \end{gather*}
  \item ("!E")
    \begin{gather*}
      \I[cut]{"\G ; \D, \D' ==> C @ {w'}"}
        {"\G ; \D ==> {! A} @ w"
         &
         \I["!L"]{"\G ; \D', {! A} @ w ==> C @ {w'}"}
           {"\G, A @ w ; \D' ==> C @ {w'}"}
        }
    \end{gather*}
  \item ("{dn} E")
    \begin{gather*}
      \I[cut]{"\G ; \D, \D' ==> [w/u] A @ w"}
        {"\G ; \D ==> now u. A @ w"
         &
         \I["{dn} L"]{"\G ; \D', now u. A @ w ==> [w/u] A @ w"}
           {\I[hyp]{"\G ; \D', [w/u] A @ w ==> [w/u] A @ w"}{}}
        }
    \end{gather*}
  \item ("at~ E")
    \begin{gather*}
      \I[cut]{"\G ; \D ==> A @ v"}
        {"\G ; \D ==> (A at v) @ w"
         &
         \I["at~ L"]{"\G ; (A at v) @ w ==> A @ v"}
           {\I[init*]{"\G ; A @ v ==> A @ v"}{}}
        }
    \end{gather*}
  \end{ecom}
\end{proof}

\begin{cor}[Consistency of \hyll]
  There is no proof of ". ; . |-nd zero @ w".
\end{cor}

\begin{proof}
  Suppose ". ; . |-nd zero @ w" is derivable. Then, by the completeness and cut-admissibility
  theorems on the sequent calculus, ". ; . ==> zero @ w" must have a cut-free proof. 
  But, we can see by simple inspection that there can be no cut-free proof of 
  ". ; . ==> zero @ w", as this sequent cannot be the conclusion of any rule of inference 
  in the sequent calculus. 
  Therefore, ". ; . |-nd zero @ w" is not derivable.
\end{proof}

\subsection{Connection to IS5}
\label{sec:proofs.is5}

\begin{thm}[\hyll is intuitionistic S5]
  The following sequent is derivable: ". ; dia A @ w ==> box dia A @ w".
\end{thm}

\begin{proof}
  \begin{gather*}
    \I[defn]{". ; dia A @ w ==> box dia A @ w"}
      {\I["\exists L^a"]{". ; ex u. (A at u) @ w ==> all u. (ex v. (A at v) at u) @ w"}
         {\I["\forall R^b"]{". ; (A at a) @ w ==> all u. (ex v. (A at v) at u) @ w"}
            {\I["at L", "at R"]{". ; (A at a) @ w ==> (ex v. (A at v) at b) @ w"}
               {\I["\exists R"]{". ; A @ a ==> ex v. (A at v) @ b"}
                  {\I["at R"]{". ; A @ a ==> (A at a) at b"}
                     {\I[init*]{". ; A @ a ==> A @ a"}{}}
                  }
               }
            }
         }
      }
      \qedhere
  \end{gather*}
\end{proof}

\end{document}